\documentclass[11pt]{llncs}

\usepackage[a4paper,hmargin=0.85in,vmargin=0.85in]{geometry}

\usepackage[dvips]{graphicx}
\usepackage{epsfig}
\usepackage{amsfonts}
\usepackage{amsmath}
\usepackage{amssymb}
\usepackage{slashbox}
\usepackage{multirow}
\usepackage{pstricks}
\usepackage{subfigure}

\newcommand{\ep}{\varepsilon}
\newcommand{\pe}{\hspace*{\fill} $\Box$\\}
\newcommand{\peon}{\hspace*{\fill} $\Box$}
\newcommand{\prob}{\mbox{Prob\, }}

\newcommand{\ket}[1]{\ensuremath{\left|#1\right\rangle}}

\newcommand{\bof}[1]{\textbf{#1}}

\newtheorem{condition}{Condition}

\spnewtheorem{protocol}{Protocol}{\bfseries}{\itshape}

\newenvironment{condition2}[2][Condition]{\begin{trivlist}
\item[\hskip \labelsep {\bfseries #1}\hskip \labelsep {\bfseries #2}]}{\end{trivlist}}

\begin{document}

\title{\bf Quantum Cryptography Based Solely on {B}ell's Theorem}
\author{Esther H\"anggi\inst{1}\ 
\ \ \ \ Renato Renner\inst{2}\ 
\ \ \ \ Stefan Wolf\, \inst{3}
}

\institute{Computer Science Department, ETH Zurich, CH-8092 Zurich, Switzerland. \\ E-mail: esther.haenggi@inf.ethz.ch
\and
Institute for Theoretical Physics, ETH Zurich, CH-8093 Zurich, Switzerland. \\ E-mail: renner@phys.ethz.ch
\and
Computer Science Department, ETH Zurich, CH-8092 Zurich, Switzerland. \\ E-mail: wolf@inf.ethz.ch
}

\maketitle

\begin{abstract}
\noindent 
Information-theoretic key agreement is impossible to achieve 
from scratch and must be based on some --- ultimately physical --- premise. 
In 2005, Barrett, Hardy, and Kent 
showed that unconditional security can  be obtained in principle based  on the 
impossibility of faster-than-light signaling; however, their protocol is inefficient and  cannot 
tolerate any noise. While their key-distribution scheme uses quantum entanglement, its 
security only relies on the impossibility of superluminal signaling, rather than the  correctness 
and completeness of quantum theory. 
In particular, the resulting security is device independent. 
Here we introduce a new protocol which is efficient in terms of both classical and quantum 
communication, and that can tolerate noise in the quantum channel. 
We prove that it offers device-independent security under the sole assumption 
that certain non-signaling conditions are satisfied. Our main insight is that 
the XOR of a number of bits that are partially secret according to the non-signaling  conditions turns out to 
be highly secret. Note that similar statements have been well-known in classical 
contexts. Earlier results had indicated that amplification of such 
non-signaling-based privacy is impossible to achieve if the non-signaling condition 
only holds between events on Alice's and Bob's sides. Here, we show that the 
situation changes completely if such a separation is given within each of the 
laboratories. 
\end{abstract}

\section{Introduction, Motivation, and Our Result}

\subsection{Minimizing Assumptions for Information-Theoretic Key Agreement}

It is  well-established  that information-theoretic secrecy must be based on
certain premises such as {\em noise\/} in communication channels~\cite{wyner},\, \cite{csikor},\, 
\cite{maurer92}, a {\em limitation\/} on an adversary's memory~\cite{maurer},\, \cite{DziMau08},
or the uncertainty principle of quantum physics~\cite{bb84}. 
In traditional quantum key distribution, the security proof is based on
\begin{enumerate}
 \item 
the postulates of quantum physics,
\item \label{item:device}
the assumptions that the used devices transmit and operate on the specified quantum
systems, and 
\item
that Eve does not get information about the generated key out of the legitimate partners'
laboratories.
\end{enumerate}
This article is concerned with a variant of quantum key distribution
which allows the first two assumptions to be dropped, if at the same time, the third is  augmented by the 
assumption that no unauthorized information is exchanged between the legitimate laboratories. One possibility to guarantee this
is via the non-signaling postulate of relativity, if different measurement events 
are carried out in a space-like separated way. Of particular importance is  
{\em device independence\/} (i.e., dropping condition~\ref{item:device}), for two reasons. First, the 
necessity to trust the manufacturer is never satisfactory. Second, the security 
of traditional protocols for quantum key distribution relies {\em crucially\/} on the fact that
single Qbits (i.e., photons) are sent. For instance, the BB84 protocol~\cite{bb84} becomes {\em completely insecure\/} if larger 
systems, such as {\em pairs of photons}, are transmitted. With present technology, this is a significant issue. The fact that practical 
deviations from the theoretical model open the possibility of attacks has been demonstrated experimentally, see~\cite{gfkzr},\, \cite{fqtl},\, \cite{qflm},\, \cite{zfqcl},\, \cite{blackpaper}, and references therein. 

The question of {\em device-independent\/} security has been raised by Mayers and Yao in~\cite{MayersYao98}.\footnote{The work by Mayers and Yao initiated 
further investigation on how to test the correct working of quantum devices (not restricted to quantum cryptography)~\cite{dmms},\, \cite{my},\, \cite{mmmo}.} 
That such security is possible in principle follows from~\cite{bhk}; however, only a zero secret-key rate has been achieved, and in addition the classical communication 
cost is exponential. Later schemes that are robust against noise and achieve a positive key rate have been proven secure against certain 
restricted types of attacks~\cite{AcinMassarPironio},\, \cite{SGBMPA},\, \cite{AcinGisinMasanes},\, \cite{abgs}. 
The current state of the art is that security holds against arbitrary attacks, but no (quantum) correlation is introduced between subsequent measurements, see e.g.,~\cite{McKague}.

\subsection{Relativity-Based Key Distribution}

It is possible to generate a secret key assuming only  that information transmission faster than at  the speed of light is impossible. 
The basic idea, as proposed by Barrett, Hardy, and Kent~\cite{bhk}, is as follows: 
By communication over a quantum channel, two parties,  Alice and Bob, generate some 
shared entangled quantum state. They carry out measurements 
in a space-like-separated way, i.e., no signaling is possible between the measurement events.
Alice and Bob then verify the statistics of the measurement outcomes. Given that these satisfy 
certain specified properties, the privacy of the data
 follows  {\em directly\/} from the correlations 
in the resulting  data and is independent of
 whatever quantum 
systems the devices operate on.
It is not even necessary to assume  that the possibilities of what an adversary can do is limited by  quantum physics: 
The 
latter guarantees the protocol to {\em work} (i.e., leads to the expected correlations,
the occurrence of which can be verified), 
{\em but the security  is completely independent of it}. A
consequence is that protocols can be given which are secure if {\em either\/} quantum physics {\em or\/} 
relativity (or both, of course) is correct.

How is it  possible to derive secrecy directly from correlations? 
In quantum physics, this is well-known: Quantum correlations, called {\em entanglement}, 
are monogamous to some extent~\cite{terhal}:
If Alice and Bob are maximally entangled, then Eve 
factors out and is independent. However, we do  not know such an 
effect classically: If Alice and Bob have highly correlated bits, Eve can nevertheless know them. 
The point is that we have to look at the 
--- so-called {\em non-local\/} ---  input-output behavior of  {\em systems}.

\subsection{Systems, Correlations, and Non-Locality}

In order to explain  non-local correlations, we introduce the notion of a {\em two-party 
system}, defined by its joint input-output behavior $P_{XY|UV}$ (see Figure~\ref{boexli}). 
\begin{figure}[ht]
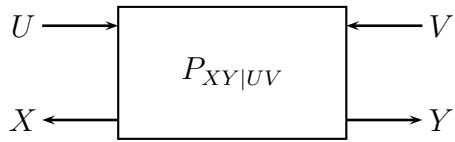

\centering
\pspicture*[](-3.5,-0.2)(3.5,1.95)
\pspolygon[](-1.5,0)(1.5,0)(1.5,1.75)(-1.5,1.75)
\rput[b]{0}(0,0.65){\large{$P_{XY|UV}$}}
\psline[linewidth=1pt]{->}(-2.5,1.5)(-1.5,1.5)
\psline[linewidth=1pt]{<-}(-2.5,0.25)(-1.5,0.25)
\psline[linewidth=1pt]{->}(2.5,1.5)(1.5,1.5)
\psline[linewidth=1pt]{<-}(2.5,0.25)(1.5,0.25)
\rput[b]{0}(-2.75,1.35){\large{$U$}}
\rput[b]{0}(2.75,1.35){\large{$V$}}
\rput[b]{0}(-2.75,0.1){\large{$X$}}
\rput[b]{0}(2.75,0.1){\large{$Y$}}
\endpspicture
\caption{A two-party {\em system}. If it does not allow for message transmission, it is called a {\em box}.}
\label{boexli}
\end{figure}

\begin{definition}
{\rm 
\label{localDef}
A {\em system} is a bi- (or more-) partite conditional probability distribution $P_{XY|UV}$. 
It is 
{\em local\/} if 
$
P_{XY|UV}=\sum_{i=1}^n{w_i P_{X|U}^i P_{Y|V}^i}
$
holds for some weights $w_i\geq 0$ and conditional distributions $P_{X|U}^i$ and $P_{Y|V}^i$, $i=1,\ldots,n$. 
A system is {\em signaling\/} if it allows for message transmission, 
i.e., it is {\em non-signaling} if 
$\sum_x P_{XY|UV}(x,y,u,v)=\sum_xP_{XY|UV}(x,y,u',v)$ for all $y,v$ (and similar with the roles of the interfaces exchanged).
We call a non-signaling system a \emph{box}.
}
\end{definition}

\noindent
Lemma~\ref{lemma:realism}  states that locality  is equivalent to the possibility that the outputs 
to alternative inputs are consistently pre-determined (see Figure~\ref{realb}).
\begin{lemma}\label{lemma:realism}
For any system $P_{XY|UV}$, where ${\cal U}$ and ${\cal V}$ are the ranges of $U$ and $V$, respectively, 
the following  conditions are equivalent:
\begin{enumerate}
\item
$P_{XY|UV}$ is local, 
\item
there exist random variables $X_u$ ($u\in{\cal U}$) and $Y_v$  ($v\in{\cal V}$)
with a joint distribution such that the marginals satisfy 
$
P_{X_uY_v}=P_{XY|U=u,V=v}
$.
\end{enumerate}
\end{lemma}
\begin{proof}
Assume first that $P_{XY|UV}$ is local, i.e.,
$
P_{XY|UV}=\sum{w_i P_{X|U}^i P_{Y|V}^i}
$.
For ${\cal U}=\{u_1,u_2,\ldots,u_m\}$ and ${\cal V}=\{v_1,v_2,\ldots,v_n\}$, define
\begin{small}
\[
P_{X_{u_1}\cdots X_{u_m}Y_{v_1}\cdots Y_{v_n}}(x_1,\ldots,x_m,y_1,\ldots,y_n):=
\sum{w_i P_{X|U=u_1}^i(x_1)\cdots P_{X|U=u_m}^i(x_m)
\cdot P_{Y|V=v_1}^i(y_1)\cdots P_{Y|V=v_n}^i(y_n)}\ .
\]
\end{small}
This distribution has the desired property.

To see the reverse direction, let $X_{u_1}\cdots X_{u_m}Y_{v_1}\cdots Y_{v_n}$ be the shared randomness~$w$. \pe
\end{proof}
\begin{figure}[ht]
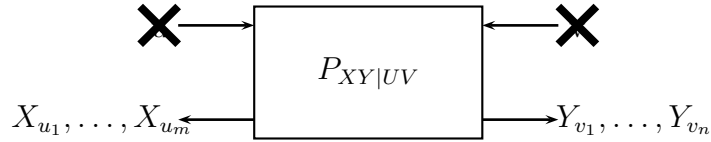

\centering
\pspicture*[](-4.7,-0.2)(4.7,1.95)
\pspolygon[](-1.5,0)(1.5,0)(1.5,1.75)(-1.5,1.75)
\rput[b]{0}(0,0.65){\large{$P_{XY|UV}$}}
\psline[linewidth=1pt]{->}(-2.5,1.5)(-1.5,1.5)
\psline[linewidth=1pt]{<-}(-2.5,0.25)(-1.5,0.25)
\psline[linewidth=1pt]{->}(2.5,1.5)(1.5,1.5)
\psline[linewidth=1pt]{<-}(2.5,0.25)(1.5,0.25)
\psline[linewidth=3pt]{-}(-3,1.75)(-2.5,1.25)
\psline[linewidth=3pt]{-}(-2.5,1.75)(-3,1.25)
\psline[linewidth=3pt]{-}(3,1.75)(2.5,1.25)
\psline[linewidth=3pt]{-}(2.5,1.75)(3,1.25)
\rput[b]{0}(-2.75,1.35){\large{u}}
\rput[b]{0}(2.75,1.35){\large{v}}
\rput[b]{0}(-3.5,0.05){\large{$X_{u_1},\ldots,X_{u_m}$}}
\rput[b]{0}(3.5,0.05){\large{$Y_{v_1},\ldots,Y_{v_n}$}}
\endpspicture
\caption{{\em Locality\/} means that alternative outputs consistently 
coexist.}
\label{realb}
\end{figure}

\noindent
We cryptographically exploit the contraposition of the statement: As soon as a system 
behaves non-locally, 
the outputs
cannot 
exist before the input is given, i.e., the measurement is actually carried out. 
In particular, these outputs cannot have been stored in the devices previously, and they
cannot be known to an adversary.

\subsection{Non-Locality Implies Secrecy}\label{nis}

In order to explain this idea more explicitly, let us consider a specific example of a system (see also Figure~\ref{fig:prbox}).

\begin{definition}{\bf \cite{pr}}
{\rm 
A {\em Popescu-Rohrlich box\/} (or {\em PR box\/} for short) is the following bipartite 
system $P_{XY|UV}$: For 
each input pair $(u,v)$, the random variable $X$  is a random bit
and we have 
\begin{equation}\label{nlbed}
\prob[X\oplus Y=U \cdot V]=1\ .
\end{equation}
}
\end{definition}

\begin{figure}[ht]
\centering
\psset{unit=0.525cm}
\pspicture*[](-2,-1)(8.5,10)
\psline[linewidth=0.5pt]{-}(0,6)(-1,7)
\rput[c]{0}(-0.25,6.75){\normalsize{$X$}}
\rput[c]{0}(-0.75,6.25){\normalsize{$Y$}}
\rput[c]{0}(-0.5,7.5){\Large{$U$}}
\rput[c]{0}(-1.5,6.5){\Large{$V$}}
\rput[c]{0}(2,7.5){\Large{$0$}}
\rput[c]{0}(6,7.5){\Large{$1$}}
\rput[c]{0}(1,6.5){\Large{$0$}}
\rput[c]{0}(3,6.5){\Large{$1$}}
\rput[c]{0}(5,6.5){\Large{$0$}}
\rput[c]{0}(7,6.5){\Large{$1$}}
\rput[c]{0}(-1.5,4.5){\Large{$0$}}
\rput[c]{0}(-1.5,1.5){\Large{$1$}}
\rput[c]{0}(-0.5,5.25){\Large{$0$}}
\rput[c]{0}(-0.5,3.75){\Large{$1$}}
\rput[c]{0}(-0.5,2.25){\Large{$0$}}
\rput[c]{0}(-0.5,0.75){\Large{$1$}}
\psline[linewidth=2pt]{-}(-1,0)(8,0)
\psline[linewidth=2pt]{-}(-1,6)(8,6)
\psline[linewidth=2pt]{-}(-1,3)(8,3)
\psline[linewidth=1pt]{-}(0,1.5)(8,1.5)
\psline[linewidth=1pt]{-}(0,4.5)(8,4.5)
\psline[linewidth=2pt]{-}(0,0)(0,7)
\psline[linewidth=2pt]{-}(8,0)(8,7)
\psline[linewidth=2pt]{-}(4,0)(4,7)
\psline[linewidth=1pt]{-}(2,0)(2,6)
\psline[linewidth=1pt]{-}(6,0)(6,6)
\rput[c]{0}(1,5.25){\Large{$\frac{1}{2}$}}
\rput[c]{0}(3,3.75){\Large{$\frac{1}{2}$}}
\rput[c]{0}(5,5.25){\Large{$\frac{1}{2}$}}
\rput[c]{0}(7,3.75){\Large{$\frac{1}{2}$}}
\rput[c]{0}(1,2.25){\Large{$\frac{1}{2}$}}
\rput[c]{0}(3,0.75){\Large{$\frac{1}{2}$}}
\rput[c]{0}(5,0.75){\Large{$\frac{1}{2}$}}
\rput[c]{0}(7,2.25){\Large{$\frac{1}{2}$}}
\rput[c]{0}(3,5.25){\Large{$0$}}
\rput[c]{0}(1,3.75){\Large{$0$}}
\rput[c]{0}(7,5.25){\Large{$0$}}
\rput[c]{0}(5,3.75){\Large{$0$}}
\rput[c]{0}(3,2.25){\Large{$0$}}
\rput[c]{0}(1,0.75){\Large{$0$}}
\rput[c]{0}(5,2.25){\Large{$0$}}
\rput[c]{0}(7,0.75){\Large{$0$}}
\endpspicture
\caption{The PR box.}
\label{fig:prbox}
\end{figure}

John Bell's theorem from 1964~\cite{bellInequality} implies that this system is indeed non-local. More precisely, 
any system that behaves like a PR box with probability greater than $75\%$ is. 
The reason is that the four conditions represented by (\ref{nlbed}) (one for each input combination) 
are contradictory, and only three can be 
satisfied at a time. 
Interestingly, 
when one is allowed to measure entangled quantum states, one  can achieve roughly $85\%$. 

The type of non-locality characterized by the PR box is often called \emph{CHSH non-locality} after~\cite{chsh}
and we will sometimes call condition (\ref{nlbed}) \emph{CHSH condition}. 

Note that the 
PR box is  non-signaling: $X$ and $Y$
separately are perfectly random bits and independent of the input pair. On the other hand, a system 
 $P_{XY|UV}$ (where all variables are bits) satisfying (\ref{nlbed}) is 
non-signaling {\em only\/} if the outputs are completely unbiased, given the input pair, i.e.,
$
P_{X|U=u,V=v}(0)=
P_{Y|U=u,V=v}(0)=1/2
$.
In other words, the output bit can neither be pre-determined, nor slightly biased.
 Assume that Alice and Bob share any kind of
 physical system, carry out space-like separated measurements (hereby excluding message transmission), 
and measure data having the statistics of a 
PR box.  The outputs must then be perfectly secret bits 
because even when conditioned on an adversary's complete information, the correlation between Alice and 
Bob must still be non-signaling and fulfill equation (\ref{nlbed}).

Unfortunately, however, the behavior of perfect 
PR boxes does not occur in nature: Quantum physics is non-local, but 
not maximally so. Can we also obtain  secret bits from 
weaker, quantum-physically achievable, non-locality? Barrett, Hardy, and Kent~\cite{bhk} have shown that 
the answer is {\em yes}. Their protocol is, however, inefficient: In order to reduce the probability that
the adversary learns a generated bit shared by Alice and Bob below $\ep$, they have to communicate 
$\Theta(1/\ep)$ Qbits. 

If we measure maximally entangled quantum states, we can get at most $85\%$-approximations to the 
PR-box's 
behavior. Fortunately, {\em any\/} non-locality implies {\em some\/} secrecy. In 
order to illustrate this, consider a system approximating a 
PR box with probability $1-\ep$ for all 
inputs. More precisely, we have 
\begin{equation}\label{nleps}
\prob[X\oplus Y=U\cdot V|U=u,V=v]=1-\ep
\end{equation}
for all $(u,v)\in\{0,1\}^2$. Then, what is the maximal possible bias 
$
p:=\prob[X=0|U=0,V=0]
$
such that the system is non-signaling?
\[
\begin{array}{|c||c|c||c|}
\hline
x & P_{X|U=u,V=v}(0) &  P_{Y|U=u,V=v}(0) & y \\
\hline
\hline
0 & p & p-\ep & 0\\\hline
0 & p & p-\ep & 1\\\hline
1 & p-2\ep & p-\ep & 0\\\hline
1 & p-2\ep & p-\ep & 1\\
\hline
\end{array}
\]
We explain the table: 
Because of (\ref{nleps}), the bias of $Y$, given $U=V=0$, must be at least $p-\ep$. Because 
of non-signaling, $X$'s bias must be $p$ as well when $V=1$, and so on. Finally, condition (\ref{nleps}) 
for $U=V=1$ implies 
$
p-\ep-(1-(p-2\ep))\leq \ep
$,
hence, 
$
p\leq 1/2+2\ep
$.
For any $\ep<1/4$, this is a non-trivial bound. (This reflects the fact that $\ep=1/4$ is 
the ``local limit.'')

Conditioned on Eve's entire information, this reads: Weak non-locality means 
weak secrecy. Can it be amplified?
 {\em Privacy amplification\/} is a concept well-known from classical~\cite{bbr},\, \cite{ill},\, \cite{BBCM95}
 and quantum~\cite{koenigrenner} cryptography, and means transforming a weakly secret string
 into a highly secret key by hashing. 
These results are not applicable with respect to  non-signaling privacy since this is a strictly stronger notion, 
i.e., the attacker has more possible courses of action.\footnote{The {\em only\/} restriction  by which
the possibilities of
such an adversary are limited  is the non-signaling condition.
Non-signaling  secrecy
has been shown achievable 
 under the additional assumption  
that the adversary can only attack each of the boxes 
 separately~\cite{AcinMassarPironio},\, \cite{SGBMPA},\, \cite{AcinGisinMasanes}.  In general, however, an adversary may of course attack 
  them jointly --- this corresponds to a coherent attack. In quantum mechanics, three types of attacks --- individual, collective, and coherent  --- are  
distinguished~\cite{BihamMor},\, \cite{BihamMor2},\, \cite{BBBGM}. In an {\em individual\/} attack, the eavesdropper attacks and measures each system identically and independently; in a {\em collective\/} attack the adversary still attacks each system identically and independently, but can make a joint measurement; finally the most 
general attack is a {\em coherent\/} attack, where no restrictions apply.} 
In~\cite{HRW08} it has been pessimistically argued that privacy amplification
of non-signaling secrecy is impossible, the problem being that certain collective attacks
exist that leave the adversary with significant information about the final key, 
however the latter is obtained from the raw key.

Fortunately, the situation changes completely when one 
assumes a non-signaling condition between the individual measurements performed \emph{within} Alice's as well as Bob's laboratories (see Figure~\ref{figure:space-like_separation}). 
This non-signaling condition could, for instance, be enforced by a space-like separation of the individual measurement events. 
In~\cite{lluis}, Masanes has shown that in this case, privacy amplification is possible in principle --- 
using as hash function a  function  chosen at random from the set of {\em all functions}.\footnote{Masanes' result implies that there exists a fixed function which can be used 
for privacy amplification, but the proof is non-constructive, i.e., the function cannot be given 
explicitly.} 
Later, he has shown that it is sufficient to consider a two-universal set of functions (this proof is included in 
\cite{mrwbcv4}, Section~IV.C).

\subsection{Main Result} 

We show that there exists a protocol for efficiently generating a secret 
key, whose security is based on non-signaling conditions only (Theorem~\ref{th:main}). 
The protocol consists of measuring $n$ copies of a 
maximally entangled state, where all $2n$ measurement events are supposed to be 
space-like 
separated. 
Our result is distinct from Masanes' in the sense that we show a \emph{single explicit function}, namely the XOR, to be a good privacy-amplification function. More 
precisely, we prove a lemma that the 
 adversary's probability of correctly predicting 
 the XOR of the 
outcomes of $n$ boxes is exponentially (in $n$) close to $1/2$ (see Lemma~\ref{lemma:xor_good_pa_fct}). This can be seen as a generalization of the well-known fact that the XOR of many partially uniform 
bits is almost uniform and may be of independent interest.
\\ 

\noindent
Since the security of our protocol, which is {\em universally composable}, is implied 
by the observed correlations alone, it is automatically {\em device-independent}. This means that 
nothing needs to be known about the internal workings of the quantum 
devices used for its implementation (such as photon sources or 
detectors) and their manufacturer need not be trusted. 
Moreover, a certain amount of noise can be tolerated: Our scheme has a positive key-generation rate 
whenever the correlations approximate 
PR boxes with an accuracy exceeding $80\%$ and the output bits are correlated with more than $98\%$ when Alice and Bob both choose to measure in the first basis (see Figure~\ref{fig:key_quantum_area}).
\begin{figure}[ht!]
\centering
\includegraphics[width=9cm]{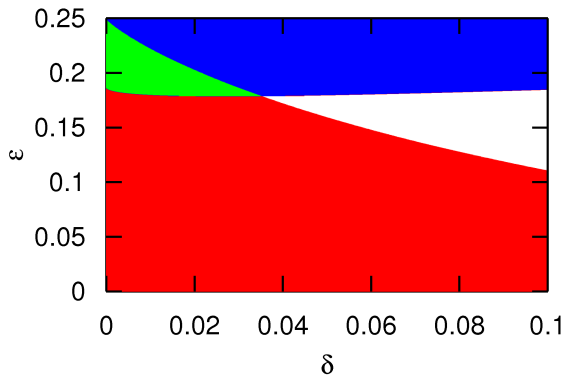}
\vspace{-0.5cm}
\caption{\label{fig:key_quantum_area} The parameter regions for which key agreement is 
possible (red), reachable by quantum mechanics (blue) and their 
intersection (green). $\ep$ is the probability of violating 
the CHSH condition (i.e., $X\oplus Y\neq U\cdot V$) for uniform inputs, and $\delta$ the probability 
of not  having the same output bits on input $(0,0)$.}
\end{figure}

\subsection{Outline}

The rest of our paper is organized as follows. In Section~\ref{sec:model}, we describe the model and the general set of possible 
strategies of a non-signaling adversary. In  Section~\ref{sec:security}, we motivate our security definition. Then we first
consider the case of a single approximation of a PR box and give
a tight bound on the adversarial knowledge on the outputs of such a box (Section~\ref{sec:limit_one_box}).
We then proceed to the general case of $n$ approximations
of a PR box  (Section~\ref{sec:xorpa}). We show that the XOR of several output bits is as secure as when an adversary attacks
each of the boxes independently and individually, hence, the XOR is a good privacy-amplification function. In Section~\ref{sec:keygeneration}, 
we show that we can use the XOR of several bits to do information reconciliation and privacy amplification such as to obtain a 
secret key. We determine the key rate, show how we can attain the region allowing for a positive key rate using quantum mechanics and finally give the 
resulting key generation protocol.

\section{Modeling Non-Signaling Adversaries}\label{sec:model}

When Alice, Bob, and Eve carry out measurements on a (joint) physical system, they can choose
their measurement settings (the inputs) and receive their respective outcomes (the outputs). It is, therefore, natural 
to model the situation by a tripartite input-output system, characterized by a conditional 
distribution $P_{XYZ|UVW}$. The question we study in the following is: 
Given a certain two-party system shared by Alice and Bob, which extensions to a three-party system, including the adversary Eve, 
are possible? And is it possible for Alice and Bob to create a secret key by interacting with their respective parts of the system and communicating over a public channel?
\begin{figure}[h]
\begin{center}
\pspicture*[](-5,-2.3)(5,1.95)
\pspolygon[](-1.5,0)(1.5,0)(1.5,1.75)(-1.5,1.75)
\rput[b]{0}(0,0.65){\large{$P_{XYZ|UVW}$}}
\psline[linewidth=1pt]{->}(-2.5,1.5)(-1.5,1.5)
\psline[linewidth=1pt]{<-}(-2.5,0.25)(-1.5,0.25)
\psline[linewidth=1pt]{->}(2.5,1.5)(1.5,1.5)
\psline[linewidth=1pt]{<-}(2.5,0.25)(1.5,0.25)
\rput[b]{0}(-2.75,1.35){\large{$U$}}
\rput[b]{0}(2.75,1.35){\large{$V$}}
\rput[b]{0}(-2.75,0.1){\large{$X$}}
\rput[b]{0}(2.75,0.1){\large{$Y$}}
\psline[linewidth=1pt]{->}(-0.675,-1)(-0.675,0)
\psline[linewidth=1pt]{<-}(0.675,-1)(0.675,0)
\rput[b]{0}(-0.6755,-1.5){\large{$W$}}
\rput[b]{0}(0.675,-1.5){\large{$Z$}}
\rput[b]{0}(-3.7,0.82){\Large{Alice}}
\rput[b]{0}(3.6,0.82){\Large{Bob}}
\rput[b]{0}(0,-2.3){\Large{Eve}}
\endpspicture
\end{center}
\vspace{-0.5cm}
\caption{The tripartite scenario including the eavesdropper.}
\label{tripartite-situation}
\end{figure}
The only condition hereby is that the entire
system must 
be {\em non-signaling\/},\footnote{In practice, the non-signaling condition can be ensured by carrying out all measurements in a space-like separated way (the system is then non-signaling by relativity theory) or, 
alternatively, by placing every partial system into a shielded laboratory. It is also a direct consequence of the assumption usually made in quantum key distribution, that the Hilbert space is the tensor product of the Hilbert spaces associated with each party.}
i.e., the input/output behavior of one side tells nothing about the input on the other side(s) (and also, dividing the ends of the box in any two subsets, the input/output behavior of one subset tells nothing about the input of the other). 
\begin{condition}{\cite{bhk}}\label{condition:ns}
The system $P_{XYZ|UVW}$ must not allow for signaling:
\begin{eqnarray}
\nonumber \sum\nolimits_x P_{XYZ|UVW}(x,y,z,u,v,w)&=&\sum\nolimits_xP_{XYZ|UVW}(x,y,z,u',v,w)\ \forall y,z,v,w\\
\nonumber \sum\nolimits_yP_{XYZ|UVW}(x,y,z,u,v,w)&=&\sum\nolimits_yP_{XYZ|UVW}(x,y,z,u,v',w)\ \forall x,z,u,w\\
\nonumber \sum\nolimits_zP_{XYZ|UVW}(x,y,z,u,v,w)&=&\sum\nolimits_zP_{XYZ|UVW}(x,y,z,u,v,w')
\ \forall x,y,u,v
\end{eqnarray}
\end{condition}
If a system is non-signaling between its interfaces, this also means that its marginal 
systems are well-defined: What happens at one of the interfaces does not depend on any 
other input. This implies that at all the interfaces, an output can always be 
provided immediately after the input has been given. 

On the other hand, we do allow for Eve to delay her choice of input (measurement) until all of Alice's and Bob's communication is finished --- in particular Eve knows the protocol of Alice and Bob and could get
information about Alice and Bob's inputs, e.g. by wiretapping messages exchanged by them during the protocol, and she can adapt her strategy. 

This tripartite scenario can be reduced to a bipartite one: Because Eve cannot signal to Alice and Bob (even together) by her choice of input, we must have
\begin{eqnarray}
\nonumber \sum\nolimits_zP_{XYZ|UVW}(x,y,z,u,v,w)=\sum\nolimits_zP_{XYZ|UVW}(x,y,z,u,v,w')=P_{XY|UV}(x,y,u,v)\ ,
\end{eqnarray}
and this is exactly the marginal box as seen by Alice and Bob. We can, therefore, see Eve's input as a choice of convex decomposition of Alice's and Bob's box and her output as indicating one part of the decomposition. 
Further, the condition that even Alice and Eve together must not be able to signal to Bob and \textit{vice versa} means that the distribution conditioned on Eve's outcome, $P^z_{XY|UV}$, must also be non-signaling between Alice and Bob. 
Informally, we can write
\begin{eqnarray}
\nonumber
\begin{array}{|ccccc|}\hline
\phantom{_{z_0}}& & & & \\
& A &  & B & \\ 
& & & &\phantom{_{z_0}} \\\hline
\end{array}
&=&
p({z_0}|w)\cdot
\begin{array}{|ccccc|}\hline
\phantom{_{z_0}}& & & & \\
& A &  & B & \\ 
& & & & _{z_0} \\\hline
\end{array}+
p({z_1}|w)\cdot
\begin{array}{|ccccc|}\hline
\phantom{_{z_0}}& & & & \\
& A &  & B & \\ 
& & & & _{z_1} \\\hline
\end{array}+ \cdots 
\end{eqnarray}
and this also covers all possibilities available to Eve. Formally, we define:
\begin{definition}
 {\rm 
A \emph{ box partition} of a given bipartite box $P_{XY|UV}$ is a family of pairs ($p^z$,$P^z_{XY|UV}$), where $p^z$ is a weight and $P^z_{XY|UV}$ is a box, such that
$P_{XY|UV}=\sum_z p^z\cdot P^z_{XY|UV} 
$.
}
\end{definition}
This definition allows us to change between the scenario of a bipartite box plus box partition and the scenario of a tripartite box, as stated in the following two lemmas.
\begin{lemma}
For any given tripartite box, $P_{XYZ|UVW}$, any input $w$ induces a box partition of the bipartite box $P_{XY|UV}$ parametrized by $z$ with $p^z:=p(z|w)$ and $P^z_{XY|UV}:=P_{XY|UV,Z=z,W=w}$.
\end{lemma}
\begin{lemma}
Given a bipartite box $P_{{XY}|{UV}}$ let $\mathcal{W}$ be a set of box partitions 
$
w=\{(p^z,P^z_{{XY}|{UV}})\}_z
$.
\ 
Then the tripartite box, where the input of the third party is $w\in\mathcal{W}$, defined by 
$P_{{XY}Z|{UV},W=w}(z):=p^z\cdot P^z_{{XY}|{UV}}
$
is non-signaling and has marginal box $P_{{XY}|{UV}}$.
\end{lemma}

Even if Alice and Bob have several input and output interfaces
we need to assume that they belong to a single big system which can be attacked by Eve as one, 
as depicted in Figure~\ref{figure:evesposs}. This also implies that 
Eve only has a single input and output variable (of any range). This scenario is analogous to Eve being able to do coherent attacks in a quantum-key-distribution protocol. 

\begin{figure}[ht]
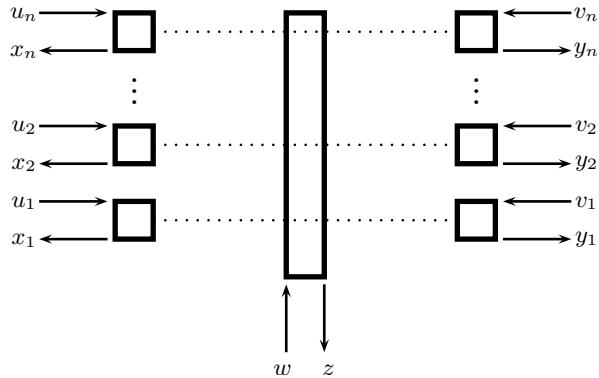

\centering
\pspicture*[](-4.5,-0.3)(4.5,5)
\rput[b]{0}(-2.25,3.3){\textbf{$\vdots$}}
\rput[b]{0}(2.25,3.3){\textbf{$\vdots$}}
\psline[linewidth=1pt]{->}(-3.5,4.5)(-2.6,4.5)
\rput[b]{0}(-3.7,4.4){$u_n$}
\psline[linewidth=1pt]{<-}(-3.5,4)(-2.6,4)
\rput[b]{0}(-3.7,3.9){$x_n$}
\pspolygon[linewidth=2pt](-2.5,4)(-2,4)(-2,4.5)(-2.5,4.5)
\pspolygon[linewidth=2pt](2.5,4)(2,4)(2,4.5)(2.5,4.5)
\psline[linewidth=1pt,linestyle=dotted]{-}(-2,4.25)(2,4.25)
\psline[linewidth=1pt]{->}(3.5,4.5)(2.6,4.5)
\rput[b]{0}(3.7,4.4){$v_n$}
\psline[linewidth=1pt]{<-}(3.5,4)(2.6,4)
\rput[b]{0}(3.7,3.9){$y_n$}
\psline[linewidth=1pt]{->}(-3.5,3)(-2.6,3)
\rput[b]{0}(-3.7,2.9){$u_{2}$}
\psline[linewidth=1pt]{<-}(-3.5,2.5)(-2.6,2.5)
\rput[b]{0}(-3.7,2.4){$x_{2}$}
\pspolygon[linewidth=2pt](-2.5,2.5)(-2,2.5)(-2,3)(-2.5,3)
\psline[linewidth=1pt]{->}(3.5,3)(2.6,3)
\rput[b]{0}(3.7,2.9){$v_{2}$}
\psline[linewidth=1pt]{<-}(3.5,2.5)(2.6,2.5)
\rput[b]{0}(3.7,2.4){$y_{2}$}
\pspolygon[linewidth=2pt](2.5,2.5)(2,2.5)(2,3)(2.5,3)
\psline[linewidth=1pt,linestyle=dotted]{-}(-2,2.75)(2,2.75)
\psline[linewidth=1pt]{->}(-3.5,2)(-2.6,2)
\rput[b]{0}(-3.7,1.9){$u_1$}
\psline[linewidth=1pt]{<-}(-3.5,1.5)(-2.6,1.5)
\rput[b]{0}(-3.7,1.4){$x_1$}
\psline[linewidth=1pt]{->}(3.5,2)(2.6,2)
\rput[b]{0}(3.7,1.9){$v_1$}
\psline[linewidth=1pt]{<-}(3.5,1.5)(2.6,1.5)
\rput[b]{0}(3.7,1.4){$y_1$}
\pspolygon[linewidth=2pt](-2.5,1.5)(-2,1.5)(-2,2)(-2.5,2)
\pspolygon[linewidth=2pt](2.5,1.5)(2,1.5)(2,2)(2.5,2)
\psline[linewidth=1pt,linestyle=dotted]{-}(-2,1.75)(2,1.75)
\pspolygon[linewidth=2pt](-0.25,1)(0.25,1)(0.25,4.5)(-0.25,4.5)
\psline[linewidth=1pt]{->}(-0.25,0)(-0.25,0.9)
\rput[b]{0}(-0.3,-0.3){$w$}
\psline[linewidth=1pt]{<-}(0.25,0)(0.25,0.9)
\rput[b]{0}(0.3,-0.3){$z$}
\endpspicture
\caption{\label{figure:evesposs}Alice and Bob share $n$  
 boxes which are independent from their viewpoint. However, Eve can attack all of them at once.}
\end{figure}

However, Alice and Bob can make sure that the non-signaling condition holds between all of their $2n$ input/output interfaces. The non-signaling condition then needs to hold even 
 \emph{given Eve's output $z$}.
We, therefore, extend Condition~\ref{condition:ns} from the tripartite to the $(2n+1)$-partite case in the obvious way and call such a system \emph{$(2n+1)$-partite non-signaling} (see Figure~\ref{figure:space-like_separation}).
\begin{figure}[ht]
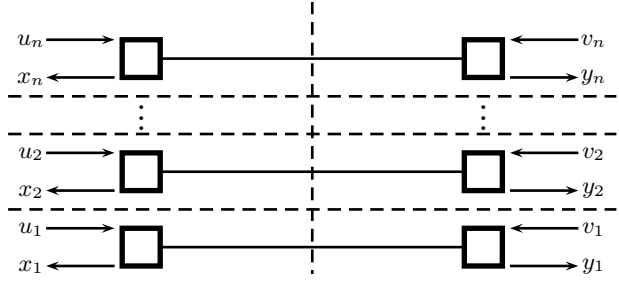

 \centering
\pspicture*[](-4.5,1.4)(4.5,5)
\rput[b]{0}(-2.25,3.3){\textbf{$\vdots$}}
\rput[b]{0}(2.25,3.3){\textbf{$\vdots$}}
\psline[linewidth=1pt,linestyle=dashed]{-}(0,1)(0,5)
\psline[linewidth=1pt]{->}(-3.5,4.5)(-2.6,4.5)
\rput[b]{0}(-3.7,4.4){$u_n$}
\psline[linewidth=1pt]{<-}(-3.5,4)(-2.6,4)
\rput[b]{0}(-3.7,3.9){$x_n$}
\pspolygon[linewidth=2pt](-2.5,4)(-2,4)(-2,4.5)(-2.5,4.5)
\pspolygon[linewidth=2pt](2.5,4)(2,4)(2,4.5)(2.5,4.5)
\psline[linewidth=1pt,linestyle=dashed]{-}(-4,3.75)(4,3.75)
\psline[linewidth=1pt]{-}(-2,4.25)(2,4.25)
\psline[linewidth=1pt]{->}(3.5,4.5)(2.6,4.5)
\rput[b]{0}(3.7,4.4){$v_n$}
\psline[linewidth=1pt]{<-}(3.5,4)(2.6,4)
\rput[b]{0}(3.7,3.9){$y_n$}
\psline[linewidth=1pt]{->}(-3.5,3)(-2.6,3)
\rput[b]{0}(-3.7,2.9){$u_{2}$}
\psline[linewidth=1pt]{<-}(-3.5,2.5)(-2.6,2.5)
\rput[b]{0}(-3.7,2.4){$x_{2}$}
\psline[linewidth=1pt,linestyle=dashed]{-}(-4,3.25)(4,3.25)
\pspolygon[linewidth=2pt](-2.5,2.5)(-2,2.5)(-2,3)(-2.5,3)
\psline[linewidth=1pt]{->}(3.5,3)(2.6,3)
\rput[b]{0}(3.7,2.9){$v_{2}$}
\psline[linewidth=1pt]{<-}(3.5,2.5)(2.6,2.5)
\rput[b]{0}(3.7,2.4){$y_{2}$}
\pspolygon[linewidth=2pt](2.5,2.5)(2,2.5)(2,3)(2.5,3)
\psline[linewidth=1pt]{-}(-2,2.75)(2,2.75)
\psline[linewidth=1pt]{->}(-3.5,2)(-2.6,2)
\rput[b]{0}(-3.7,1.9){$u_1$}
\psline[linewidth=1pt]{<-}(-3.5,1.5)(-2.6,1.5)
\rput[b]{0}(-3.7,1.4){$x_1$}
\psline[linewidth=1pt]{->}(3.5,2)(2.6,2)
\rput[b]{0}(3.7,1.9){$v_1$}
\psline[linewidth=1pt]{<-}(3.5,1.5)(2.6,1.5)
\rput[b]{0}(3.7,1.4){$y_1$}
\psline[linewidth=1pt,linestyle=dashed]{-}(-4,2.25)(4,2.25)
\pspolygon[linewidth=2pt](-2.5,1.5)(-2,1.5)(-2,2)(-2.5,2)
\pspolygon[linewidth=2pt](2.5,1.5)(2,1.5)(2,2)(2.5,2)
\psline[linewidth=1pt]{-}(-2,1.75)(2,1.75)
\endpspicture
\caption{\label{figure:space-like_separation}The dashed lines mean space-like separation.}
\end{figure}

We study the particular case where Alice and Bob share $n$ approximations of a PR box, i.e., each of the $2n$ input/output interfaces takes one bit input and gives one bit 
output.\footnote{We will write $U$ for the random bit denoting Alice's input, bold-face letters $\bof{U}$ will denote a $n$-bit 
random variable (i.e., an $n$-bit vector), $U_i$ a single random bit in this $n$-bit string and lowercase letters the value that the random 
variable has taken. 
A similar notation is used for Alice's output $X$ and Bob's input and output $V$ and $Y$. No assumption is made about the range of Eve's 
input/output variables $W$ and $Z$.} 
Note that we assume that the boxes Alice and Bob share were {\em created\/} by Eve. We can, therefore, not make any assumption about their form (i.e., the probability distribution describing them). In particular, they need not be independent approximations of PR boxes. However, Alice and Bob can test the properties of their systems and can ensure that 
the non-signaling condition holds between all $2n$ ends and even \emph{given Eve's output $z$}, i.e., $P^z_{XY|UV}$ must not allow for signaling between any of the $2n$ 
input/output bit pairs shared between Alice and Bob.
We restate the condition, under which we will prove security:
\begin{condition2}{\ref{condition:ns}'}\label{condition:ns2}
The system $P_{\bof{XY}Z|\bof{UV}W}$ must not allow for signaling between any of the $2n+1$ marginal systems:
\begin{eqnarray}
\nonumber \sum
_{x_i} P_{\bof{XY}Z|\bof{UV}W}(\bof{x},\bof{y},z,\bof{u}\backslash u_i, u_i,\bof{v},w)&=&\sum
_{x_i}P_{\bof{XY}Z|\bof{UV}W}(\bof{x},\bof{y},z,\bof{u}\backslash u_i, u'_i,\bof{v},w)\ \forall \bof{x}\backslash x_i,\bof{y},z,,\bof{u}\backslash u_i,\bof{v},w\\
\nonumber \sum
_{y_i} P_{\bof{XY}Z|\bof{UV}W}(\bof{x},\bof{y},z,\bof{u},\bof{v}\backslash v_i, v_i,w)&=&\sum
_{y_i}P_{\bof{XY}Z|\bof{UV}W}(\bof{x},\bof{y},z,\bof{u},\bof{v}\backslash v_i, v'_i,w)\ \forall \bof{x},\bof{y}\backslash y_i,z,,\bof{u},\bof{v}\backslash v_i,w\\
\nonumber \sum_zP_{\bof{XY}Z|\bof{UV}W}(\bof{x},\bof{y},z,\bof{u},\bof{v},w)&=&\sum_zP_{\bof{XY}Z|\bof{UV}W}(\bof{x},\bof{y},z,\bof{u},\bof{v},w')
\ \forall \bof{x},\bof{y},\bof{u},\bof{v}\ ,
\end{eqnarray}
where we used the notation $\bof{x}\backslash x_i$ to abbreviate $x_1,\ldots,x_{i-1},x_{i+1},\ldots x_n$, i.e., all $x_j$ for which $j\neq i$.
\end{condition2}
Note that the above conditions imply the non-signaling condition between any partition of the input/output interfaces. An explicit proof of this is given in Appendix~\ref{sec:imply_ns}.

\section{Security Definition}\label{sec:security}

\subsection{Indistinguishability}

We define security in the context of \emph{random systems}~\cite{Maurer02}. 
A \emph{system} is an object taking inputs and giving outputs --- such as, for example, a box or several boxes.
The different interfaces, number of interactions, and, if there is, the time-wise ordering of these inputs and outputs is described in the definition of the system. 
\begin{figure}[ht!]
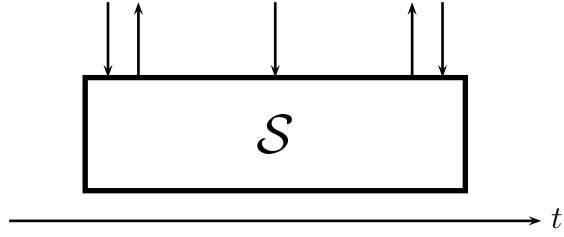

\begin{center}
\pspicture*[](-4,0)(4,3)
\psline[linewidth=1pt]{->}(-2.2,3)(-2.2,2)
\psline[linewidth=1pt]{<-}(-1.8,3)(-1.8,2)
\psline[linewidth=1pt]{->}(0,3)(0,2)
\psline[linewidth=1pt]{<-}(1.8,3)(1.8,2)
\psline[linewidth=1pt]{->}(2.2,3)(2.2,2)
\rput[b]{0}(0,1){\huge{$\mathcal{S}$}}
\pspolygon[linewidth=2pt](-2.5,0.5)(2.5,0.5)(2.5,2)(-2.5,2)
\psline[linewidth=1pt]{->}(-3.5,0.1)(3.5,0.1)
\rput[b]{0}(3.7,0){\large{$t$}}
\endpspicture
\end{center}
\vspace{-0.5cm}
\caption{A system.}
\end{figure}

The closeness of two systems $\mathcal{S}_0$ and $\mathcal{S}_1$ can be measured by introducing a so-called \emph{distinguisher}. A distinguisher $\mathcal{D}$ is itself a system and it has the same interfaces as the system $\mathcal{S}_0$, with the only difference that wherever $\mathcal{S}_0$ takes an input, $\mathcal{D}$ gives an output and \emph{vice versa}. In addition, $\mathcal{D}$ has an extra output. The distinguisher $\mathcal{D}$ has access to \emph{all} interfaces of $\mathcal{S}_0$, even though these interfaces might not be in the same location when the protocol is executed (for example, one of the interfaces might be the one seen by Alice, while the other is the one seen by Eve). 

\begin{figure}[ht!]
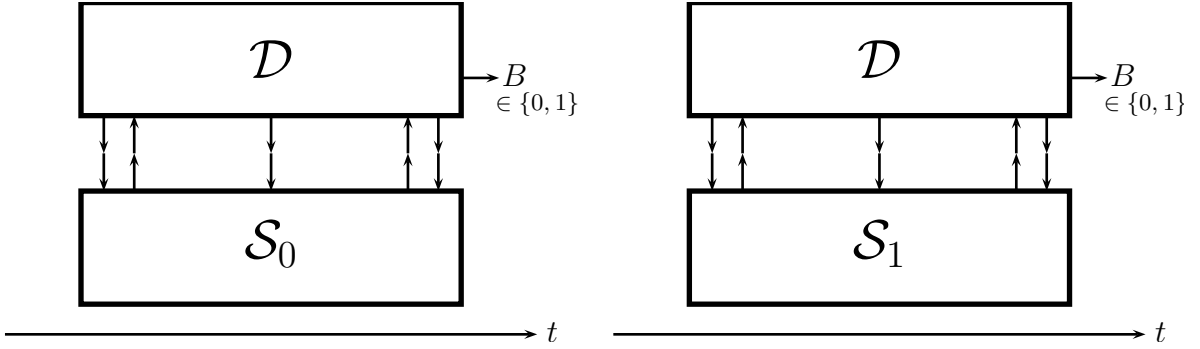

\begin{center}
\pspicture*[](-8,0)(8,4.5)
\rput[b]{0}(-4,0){
\psline[linewidth=1pt]{<-}(-2.2,2.5)(-2.2,3)
\psline[linewidth=1pt]{->}(-1.8,2.5)(-1.8,3)
\psline[linewidth=1pt]{<-}(0,2.5)(0,3)
\psline[linewidth=1pt]{->}(1.8,2.5)(1.8,3)
\psline[linewidth=1pt]{<-}(2.2,2.5)(2.2,3)
\psline[linewidth=1pt]{->}(2.5,3.5)(3,3.5)
\rput[b]{0}(3.2,3.4){\large{$B$}}
\rput[b]{0}(3.5,3){{$\in\{0,1\}$}}
\rput[b]{0}(0,3.5){\huge{$\mathcal{D}$}}
\pspolygon[linewidth=2pt](-2.5,3)(2.5,3)(2.5,4.5)(-2.5,4.5)
\psline[linewidth=1pt]{->}(-2.2,2.5)(-2.2,2)
\psline[linewidth=1pt]{<-}(-1.8,2.5)(-1.8,2)
\psline[linewidth=1pt]{->}(0,2.5)(0,2)
\psline[linewidth=1pt]{<-}(1.8,2.5)(1.8,2)
\psline[linewidth=1pt]{->}(2.2,2.5)(2.2,2)
\rput[b]{0}(0,1){\huge{$\mathcal{S}_0$}}
\pspolygon[linewidth=2pt](-2.5,0.5)(2.5,0.5)(2.5,2)(-2.5,2)
\psline[linewidth=1pt]{->}(-3.5,0.1)(3.5,0.1)
\rput[b]{0}(3.7,0){\large{$t$}}
}
\rput[b]{0}(4,0){
\psline[linewidth=1pt]{<-}(-2.2,2.5)(-2.2,3)
\psline[linewidth=1pt]{->}(-1.8,2.5)(-1.8,3)
\psline[linewidth=1pt]{<-}(0,2.5)(0,3)
\psline[linewidth=1pt]{->}(1.8,2.5)(1.8,3)
\psline[linewidth=1pt]{<-}(2.2,2.5)(2.2,3)
\psline[linewidth=1pt]{->}(2.5,3.5)(3,3.5)
\rput[b]{0}(3.2,3.4){\large{$B$}}
\rput[b]{0}(3.5,3){{$\in\{0,1\}$}}
\rput[b]{0}(0,3.5){\huge{$\mathcal{D}$}}
\pspolygon[linewidth=2pt](-2.5,3)(2.5,3)(2.5,4.5)(-2.5,4.5)
\psline[linewidth=1pt]{->}(-2.2,2.5)(-2.2,2)
\psline[linewidth=1pt]{<-}(-1.8,2.5)(-1.8,2)
\psline[linewidth=1pt]{->}(0,2.5)(0,2)
\psline[linewidth=1pt]{<-}(1.8,2.5)(1.8,2)
\psline[linewidth=1pt]{->}(2.2,2.5)(2.2,2)
\rput[b]{0}(0,1){\huge{$\mathcal{S}_1$}}
\pspolygon[linewidth=2pt](-2.5,0.5)(2.5,0.5)(2.5,2)(-2.5,2)
\psline[linewidth=1pt]{->}(-3.5,0.1)(3.5,0.1)
\rput[b]{0}(3.7,0){\large{$t$}}
}
\endpspicture
\end{center}
\vspace{-0.5cm}
\caption{\label{fig:distinguisher} The distinguisher}
\end{figure}
Now consider the following game: the distinguisher~$\mathcal{D}$ is given one out of two systems at random --- either 
$\mathcal{S}_0$ or $\mathcal{S}_1$ --- but the distinguisher does not know which one. It then has to interact with 
the system and output a bit $B$ at the end, guessing which system it has interacted with. The \emph{distinguishing 
advantage between system $\mathcal{S}_0$ and $\mathcal{S}_1$} is the maximum  guessing advantage any distinguisher 
can have in this game (see Figure~\ref{fig:distinguisher}).
\begin{definition}
The \emph{distinguishing advantage between two systems $\mathcal{S}_0$ and $\mathcal{S}_1$ }is 
\begin{eqnarray}
 \nonumber \delta(\mathcal{S}_0, \mathcal{S}_1)&=& \max_{\mathcal{D}}[P(B=1|\mathcal{S}=\mathcal{S}_0)-P(B=1|\mathcal{S}=\mathcal{S}_1)].
\end{eqnarray}
Two systems $\mathcal{S}_0$ and $\mathcal{S}_1$ are called $\epsilon$-indistinguishable if $\delta(\mathcal{S}_0, \mathcal{S}_1)\leq \epsilon$.
\end{definition}

The probability of any event $\mathcal{E}$ when the distinguisher $\mathcal{D}$ is interacting with $\mathcal{S}_0$ or $\mathcal{S}_1$ cannot differ by more than this quantity. 
\begin{lemma}\label{lemma:event}
Assume two $\epsilon$-indistinguishable systems $\mathcal{S}_0$ and $\mathcal{S}_1$.
Denote by $P(\mathcal{E}|\mathcal{S}_0,\mathcal{D})$ the probability of an event $\mathcal{E}$, defined by any 
of the input and output variables, given the distinguisher~$\mathcal{D}$ is interacting with the system $\mathcal{S}_0$. Then
\begin{eqnarray}
 \nonumber P(\mathcal{E}|\mathcal{S}_0,\mathcal{D})&\leq & P(\mathcal{E}|\mathcal{S}_1,\mathcal{D})+ \epsilon
\end{eqnarray}
\end{lemma}
\begin{proof}
Assume $P(\mathcal{E}|\mathcal{S}_0,\mathcal{D})> P(\mathcal{E}|\mathcal{S}_1,\mathcal{D})+ \epsilon$ and define 
the distinguisher $\mathcal{D}$ such that it outputs $B=0$ whenever the event $\mathcal{E}$ has happened and whenever 
$\mathcal{E}$ has not happened it outputs $B=1$. Then this distinguisher reaches a distinguishing advantage of $\delta(\mathcal{S}_0, \mathcal{S}_1)>\epsilon$ contradicting the assumption that the two systems are $\epsilon$-indistinguishable. \pe
\end{proof}

\subsection{Security of a Key}

The security of a cryptographic primitive can be measured by the distance of this system from 
an \emph{ideal} system, which is secure by definition. For example, in the case of key distribution 
the ideal system is the one which outputs a uniform and random key (bit string) at one end and for which all other 
input/output interfaces are completely independent of this first interface. This key is secure by construction. 
If the \emph{real} key distribution protocol is $\epsilon$-indistinguishable from the ideal one, then by Lemma~\ref{lemma:event} the 
key obtained from the real system needs to be secure except with probability $\epsilon$. 
This is because the probability that an adversary has knowledge about the key is $0$ in the ideal case. 

\begin{figure}[ht!]
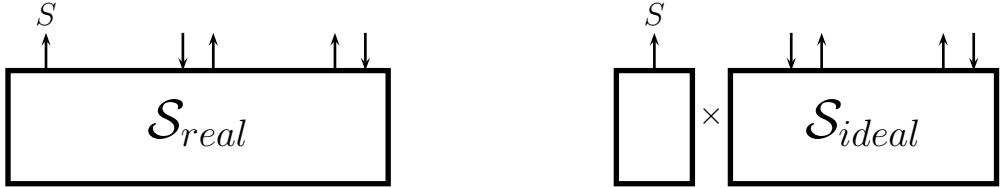

\begin{center}
\pspicture*[](-8,0)(8,3)
\rput[b]{0}(4,0){
\psline[linewidth=1pt]{<-}(-2.0,2.5)(-2.0,2)
\pspolygon[linewidth=2pt](-2.5,0.5)(-1.5,0.5)(-1.5,2)(-2.5,2)
\rput[b]{0}(-2,2.6){\large{$S$}}
\rput[b]{0}(-1.25,1.25){\large{$\times$}}
\psline[linewidth=1pt]{->}(-0.2,2.5)(-0.2,2)
\psline[linewidth=1pt]{<-}(0.2,2.5)(0.2,2)
\psline[linewidth=1pt]{<-}(1.8,2.5)(1.8,2)
\psline[linewidth=1pt]{->}(2.2,2.5)(2.2,2)
\rput[b]{0}(0.75,1){\huge{$\mathcal{S}_{ideal}$}}
\pspolygon[linewidth=2pt](-1,0.5)(2.5,0.5)(2.5,2)(-1,2)
}
\rput[b]{0}(-4,0){
\psline[linewidth=1pt]{<-}(-2.0,2.5)(-2.0,2)
\rput[b]{0}(-2,2.6){\large{$S$}}
\psline[linewidth=1pt]{->}(-0.2,2.5)(-0.2,2)
\psline[linewidth=1pt]{<-}(0.2,2.5)(0.2,2)
\psline[linewidth=1pt]{<-}(1.8,2.5)(1.8,2)
\psline[linewidth=1pt]{->}(2.2,2.5)(2.2,2)
\rput[b]{0}(0,1){\huge{$\mathcal{S}_{real}$}}
\pspolygon[linewidth=2pt](-2.5,0.5)(2.5,0.5)(2.5,2)(-2.5,2)
}
\endpspicture
\end{center}
\vspace{-0.5cm}
\caption{The real and ideal system for the case of key distribution.}
\end{figure}

\begin{definition}
A key $S$ is \emph{$\epsilon$-secure} if the system outputting $S$ is $\epsilon$-indistinguishable from an ideal system which outputs a uniform random variable $S$ and for which all other input/output interfaces are completely independent of the random variable $S$. 
\end{definition}

This definition implies that the resulting security is \emph{universally composable}~\cite{pw,bpw,canetti}. In fact, 
assume by contradiction that there exists any way of using the key (or any other part of the system which generates the key) such that 
the result is insecure, i.e., distinguishable with probability larger than $\epsilon$ from the ideal system. 
This process could then be used to distinguish the key 
generation scheme from an ideal one with probability larger 
than $\epsilon$, which is impossible by definition.

\subsection{Security of Our Key Agreement Protocol}

The system we consider (see Figure~\ref{figure:our_system_physical}) is the one where 
Alice and Bob share a public authenticated channel plus a quantum state (modeled as a box). 
Eve can wire-tap the public channel and choose an input on her part of the box and obtain an output (i.e., measure her part of the quantum state). 
Similar to the quantum case, it is no advantage for Eve to make several box partitions (measurements) instead of a single one, as the same information can be obtained by making a refined box partition of the initial box. Without loss of generality, we can, therefore, assume that Eve gives a single input
 at the end (after all communication between Alice and Bob is finished).
In our scenario, Eve, therefore, obtains all the 
communication exchanged over the public channel $Q$, can 
then choose the input to her box $W$ (which can depend on 
$Q$) and finally obtains the outcome of the box $Z$. 
\begin{figure}[ht!]
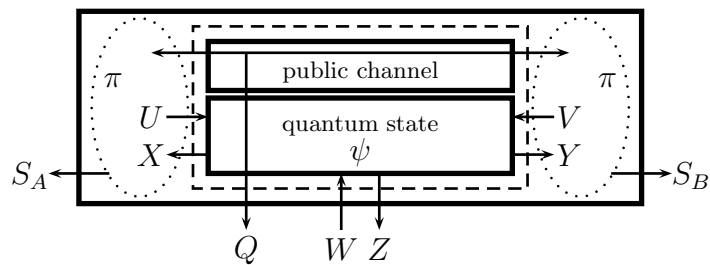

\begin{center}
\pspicture*[](-8,0)(8,3.5)
\psline[linewidth=1pt]{<->}(-2.75,2.85)(2.75,2.85)
\rput[b]{0}(0,2.45){public channel}
\pspolygon[linewidth=2pt](-2,3)(2,3)(2,2.35)(-2,2.35)
\pspolygon[linewidth=2pt](-2,1.25)(2,1.25)(2,2.25)(-2,2.25)
\psline[linewidth=1pt]{<-}(-2.55,1.5)(-2,1.5)
\rput[b]{0}(-2.75,1.35){\large{$X$}}
\psline[linewidth=1pt]{->}(-2.55,2)(-2,2)
\rput[b]{0}(-2.75,1.85){\large{$U$}}
\psline[linewidth=1pt]{<-}(2.55,1.5)(2,1.5)
\rput[b]{0}(2.75,1.35){\large{$Y$}}
\psline[linewidth=1pt]{->}(2.55,2)(2,2)
\rput[b]{0}(2.75,1.85){\large{$V$}}
\psline[linewidth=1pt]{->}(-0.25,0.5)(-0.25,1.25)
\rput[b]{0}(-0.25,0.1){\large{$W$}}
\psline[linewidth=1pt]{<-}(0.25,0.5)(0.25,1.25)
\rput[b]{0}(0.25,0.1){\large{$Z$}}
\rput[b]{0}(0,1.75){quantum state}
\rput[b]{0}(0,1.35){\large{$\psi$}}
\psline[linewidth=1pt]{<-}(-1.5,0.5)(-1.5,2.85)
\rput[b]{0}(-1.5,0.05){\large{$Q$}}
\pspolygon[linewidth=1pt,linestyle=dashed](-2.2,1.05)(2.2,1.05)(2.2,3.2)(-2.2,3.2) 
\psellipse[linewidth=1pt,linestyle=dotted](-2.9,2.125)(0.65,1.2)
\psellipse[linewidth=1pt,linestyle=dotted](2.9,2.125)(0.65,1.2)
\pspolygon[linewidth=2pt](-3.7,0.85)(3.7,0.85)(3.7,3.4)(-3.7,3.4) 
\rput[b]{0}(-3.25,2.4){\large{$\pi$}}
\rput[b]{0}(3.25,2.4){\large{$\pi$}}
\psline[linewidth=1pt]{<-}(-4.1,1.25)(-3.3,1.25)
\rput[b]{0}(-4.35,1.05){\large{$S_A$}}
\psline[linewidth=1pt]{<-}(4.1,1.25)(3.3,1.25)
\rput[b]{0}(4.35,1.05){\large{$S_B$}}
\endpspicture
\end{center}
\caption{\label{figure:our_system_physical} Our system. Alice and Bob share a public authentic channel and a quantum state. When they apply a protocol $\pi$ to obtain a key, all this can together be modeled as a system.}
\end{figure}
If Alice and Bob apply a protocol $\pi$ to the inputs and 
outputs of their boxes and the information exchanged over 
the public channel to obtain a key, this protocol can also be 
included in the system. The new system now outputs the key 
$S_A$ on Alice's and $S_B$ on Bob's side. Obviously Eve's 
possibilities to interact with this system has not changed. 
We will, therefore, need to bound the distance between this system and the ideal 
system.\footnote{Note that we can consider the distance of 
$S_A$ from an ideal key and the distance between $S_A$ and 
$S_B$ (probability of the keys to be unequal) separately. By the triangle inequality, the distance of the total real system from 
the ideal system is at most the sum of the two.} 

\begin{figure}[ht!]
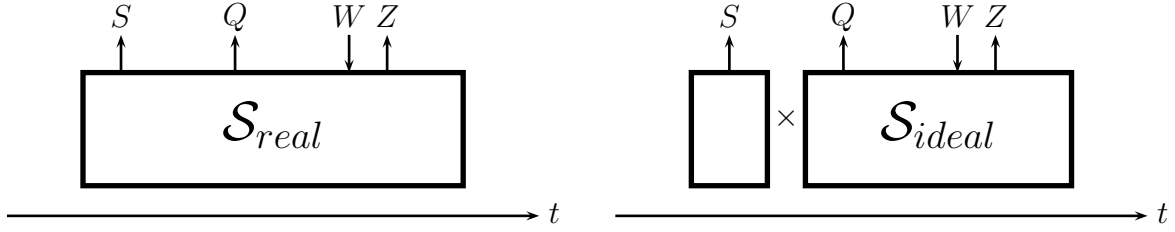

\begin{center}
\pspicture*[](-8,0)(8,3)
\rput[b]{0}(-4,0){
\psline[linewidth=1pt]{<-}(-2.0,2.5)(-2.0,2)
\rput[b]{0}(-2,2.6){\large{$S$}}
\psline[linewidth=1pt]{<-}(-0.5,2.5)(-0.5,2)
\rput[b]{0}(-0.5,2.55){\large{$Q$}}
\psline[linewidth=1pt]{->}(1.0,2.5)(1.0,2)
\rput[b]{0}(1,2.6){\large{$W$}}
\psline[linewidth=1pt]{<-}(1.5,2.5)(1.5,2)
\rput[b]{0}(1.5,2.6){\large{$Z$}}
\rput[b]{0}(0,1){\huge{$\mathcal{S}_{real}$}}
\pspolygon[linewidth=2pt](-2.5,0.5)(2.5,0.5)(2.5,2)(-2.5,2)
\psline[linewidth=1pt]{->}(-3.5,0.1)(3.5,0.1)
\rput[b]{0}(3.7,0){\large{$t$}}
}
\rput[b]{0}(4,0){
\psline[linewidth=1pt]{<-}(-2.0,2.5)(-2.0,2)
\rput[b]{0}(-2,2.6){\large{$S$}}
\psline[linewidth=1pt]{<-}(-0.5,2.5)(-0.5,2)
\rput[b]{0}(-0.5,2.55){\large{$Q$}}
\psline[linewidth=1pt]{->}(1.0,2.5)(1.0,2)
\rput[b]{0}(1,2.6){\large{$W$}}
\psline[linewidth=1pt]{<-}(1.5,2.5)(1.5,2)
\rput[b]{0}(1.5,2.6){\large{$Z$}}
\rput[b]{0}(-1.25,1.25){\large{$\times$}}
\rput[b]{0}(0.75,1){\huge{$\mathcal{S}_{ideal}$}}
\pspolygon[linewidth=2pt](-2.5,0.5)(-1.5,0.5)(-1.5,2)(-2.5,2)
\pspolygon[linewidth=2pt](-1,0.5)(2.5,0.5)(2.5,2)(-1,2)
\psline[linewidth=1pt]{->}(-3.5,0.1)(3.5,0.1)
\rput[b]{0}(3.7,0){\large{$t$}}
}
\endpspicture
\end{center}
\vspace{-0.5cm}
\caption{\label{figure:our_system}Our system. The distribution of the random variable $S$ in the ideal case is such that $P_S(s)=1/|\mathcal{S}|$.}
\end{figure}

The following corollary is a direct consequence\footnote{For the formal proof it is useful to note that instead of a box taking
input $W$, we can consider a box giving outputs indexed by $w$, $Z_w$, of which one is selected. This reflects that the box 
considered is non-signaling.} of the definitions of the systems in Figure~\ref{figure:our_system} and the distinguishing advantage.
\begin{corollary}
\label{def:dist_advantage}
Assume a key $S$ generated by a system as given in Figure~\ref{figure:our_system}. Then
\begin{eqnarray}
\nonumber 
\delta(\mathcal{S}_{real},\mathcal{S}_{ideal})&=&1/2\cdot
\sum_{s,q}  \max_{w} \sum_{z} P_{Z,Q|W=w}(z,q)\cdot|P_{S|Z=z,Q=q,W=w}(s)-P_U|,
\end{eqnarray}
where $w$ is chosen such as to maximize this quantity and $P_U:=1/|\mathcal{S}|$.
\end{corollary}
This quantity will be the one that is relevant for our security definition and because it corresponds to 
the distance from uniform of the key from the eavesdropper's point of view, we will in the following 
call it \emph{the distance from uniform of $S$ given $Z(W)$ and $Q$}, where we write $Z(W)$ because the eavesdropper 
can choose the input adaptively and the choice of input changes the output distribution. 
\begin{definition}
The \emph{distance from uniform of $S$ given $Z(W)$ and $Q$} is 
\begin{eqnarray}
\nonumber
 d(S|Z(W),Q)
&=&1/2\cdot
\sum_{s,q}  \max_{w} \sum_{z} P_{Z,Q|W=w}(z,q)\cdot|P_{S|Z=z,Q=q,W=w}(s)-P_U|
\ .
\end{eqnarray}
\end{definition}

\section{Secrecy from a Single Box}\label{sec:limit_one_box}

Let us take a closer look at the simple case where the protocol $\pi$ directly takes the output of an imperfect 
PR box as a key. More explicitly, Alice and Bob share an imperfect 
PR box --- one that fulfills $P(X\oplus Y=U\cdot V)=1-\ep$ for uniform inputs. Alice and Bob use the box giving a random input and obtain an output. Then they announce their inputs over the public authentic channel, i.e., 
$Q:=(U=u,V=v)$.\footnote{We will, in a certain abuse of notation, allow  $Q$ to consist of both random variables and 
events that a random variable takes a given value. 
In case of such events $\{U=u\}$,  this means that the 
distance from uniform will hold \emph{given this specific value 
$u$},  whereas taking the expectation  over $Q$ will correspond to taking the expectation  over all the ``free'' random variables 
contained in $Q$.} 
We will show in this section that Eve can get some 
knowledge about Alice's outcome $X$ depending on $\ep$, but 
the distance from uniform from her point of view is limited  
by $2\ep$ (assuming she gets to know the input). 
\begin{lemma}\label{lemma:guessing_probability_single_box}
Assume a tripartite box $P_{XYZ|UVW}$ such that the marginal $P_{XY|UV}$ is a non-local box 
with 
$1/4\cdot \sum_{x\oplus y=u\cdot v}P_{XY|UV}(x,y,u,v)=1-\ep$ 
and $Q:=(U=u,V=v)$. Then 
\begin{eqnarray}
 \nonumber 
d(X|Z(W),Q)&\leq& 2\ep\ .
\end{eqnarray}
\end{lemma}
\begin{proof}
Consider w.l.o.g. the case $X=0$. 
First we generalize the table from Section~\ref{nis}  to the case where $P(X\oplus Y\neq U\cdot V)=\ep$ on average  
(and it is not necessarily $\ep$ for every single input).
We call $\ep_i$ the probability not to fulfill the CHSH condition ($X\oplus Y\neq U\cdot V$) for the inputs $\{(0,0),(0,1),(1,0),(1,1)\}$ respectively. Suppose w.l.o.g. that the input was $(0,0)$, so $X$ should be maximally biased for this input.
\[
\begin{array}{|c||c|c||c|}
\hline
u & P_{X|U=u,V=v}(0) &  P_{Y|U=u,V=v}(0) & v \\
\hline
\hline
0 & p & p-\ep_1 & 0\\\hline
0 & p & p-\ep_2 & 1\\\hline
1 & p-\ep_1-\ep_3 & p-\ep_1 & 0\\\hline
1 & p-\ep_1-\ep_3 & p-\ep_2 & 1\\
\hline
\end{array}
\]
Because $P[X\oplus Y\neq U\cdot V|U,V=0,0]=\ep_1$, the bias of $Y$, given $U=V=0$, must be at least $p-\ep_1$. Because 
of non-signaling, $X$'s bias must be $p$ as well when $V=1$, and so on. Finally, 
$P[X\oplus Y\neq U\cdot V|U,V=1,1]=\ep_4$ implies
$
p-\ep_2-(1-(p-\ep_1-\ep_3))\leq \ep_4
$,
hence, 
$
p\leq 1/2+1/2\sum_i\ep_i=1/2+2\ep
$.
Now consider a box partition of $P_{XY|UV}$ parametrized by $z$. 
Let $\ep_z$ denote 
of the box given $Z=z$, i.e., $\ep_z=1/4\cdot \sum_i \ep_{i,z}$. Because this box must still be non-signaling, 
the bias of $X$ given $Z=z$, $U=u$ and $V=v$ is at most $2\ep_z$ by the above argument. 
However, because $P_{XY|UV}=\sum_z p^z\cdot P^z_{XY|UV}$, we also have $\ep=\sum_z p^z\cdot \ep_z$ and because this further
holds for all values of $X$, $
d(X|Z(W),Q)\leq 
 \sum_z p^z\cdot 2\ep_z
=2\ep$. 
 \pe
\end{proof}
\begin{remark}
Note that there exists a box partition which reaches this bound, and that can be found through a straight-forward maximization. The explicit calculations are given in Appendix~\ref{sec:single_box}. 
\end{remark}
Boxes $P_{XY|UV}$ that approximate a 
PR box with error  $\ep\in [0,0.25)$ are  \emph{non-local}. We see that for any non-local box, Eve cannot obtain perfect knowledge about Alice's output bit, and the box, therefore, contains some secrecy. 

\section{Privacy Amplification}\label{sec:xorpa}

In the following, we consider the case where Alice and Bob share $n$ imperfect PR boxes and the key is obtained by taking the XOR of all $n$ output bits. 
We will show in this section, that taking the XOR of the outputs of several boxes is a good privacy amplification function. At first, we will assume that the boxes as seen by Alice and Bob are $n$ independent and unbiased (i.e., for all inputs, the outputs $X$ and $Y$ are equally likely to be $0$ or $1$) boxes each with an associated error $\ep_i$ (which is the same for all inputs).
Then, we will show that this also holds for the case when 
Alice and Bob share a convex combination of independent 
unbiased boxes. Indeed, Alice and Bob can apply a local 
mapping to their inputs and outputs to obtain a marginal box 
that is the convex combination of several independent and 
unbiased boxes and, therefore, enforce this 
situation~\cite{mag,mrwbc}. The details of this local 
mapping --- called depolarization --- are described in 
Appendix~\ref{sec:depol}. 
Finally, we completely remove the criterion of 
independence and show that the distance from uniform of the 
XOR of the outcomes of any $(2n+1)$-non-signaling system 
having binary inputs and outputs cannot be larger than what 
could be obtained from its ``depolarized'' version. 

The main result of this section will be the following lemma.  
\begin{lemma}\label{lemma:xor_good_pa_fct}
Assume a $(2n+1)$-partite box $P_{\bof{XY}Z|\bof{UV}W}$ such that the marginal $P_{\bof{XY}|\bof{UV}}$ corresponds to $n$ independent and  unbiased non-local boxes each with an associated error $\ep_i$. Assume $f(\bof{X}):=\bigoplus_i X_i$ and $Q:=(\bof{U}=\bof{u},\bof{V}=\bof{v},F=\bigoplus)$. Then
\begin{eqnarray}
d(f(\bof{X})|Z(W),Q)&\leq& 1/2\cdot \prod\nolimits_i(4\ep_i)\ \ \left( \leq 1/2\cdot (4\ep)^n\right)\ ,
\end{eqnarray}
where $\ep=\frac{1}{n}\sum_i \ep_i$.
\end{lemma}
Note that this value can easily be reached by attacking each box independently, such as given in  Appendix~\ref{sec:single_box}, and this bound is, therefore, tight. 

For the proof of Lemma~\ref{lemma:xor_good_pa_fct} we will proceed in several steps. First, we show that the problem of finding the maximum distance from uniform of the XOR of several output bits can be cast as a linear optimization problem. Then, we show that this linear program describing $n$ boxes can be seen as the $n$-wise tensor product of the linear program describing 
a single box --- this is the crucial step. 
By using the product form of the linear program we can then show that there exists a dual feasible solution --- i.e., an upper-bound on the distance from uniform --- reaching the above value. 
\\

First we note that the maximal possible non-uniformity of the XOR of the output bits can be obtained by a box partition with only two outputs, $0$ and $1$.
\begin{lemma}\label{lemma:p_half}
Assume there exists a box partition with $d(\bigoplus_i X_i|Z'(W),Q)$. Then there exists a box partition with the same distance from uniform with $Z\in \{0,1\}$.
\end{lemma}
\begin{proof}
Assume that the box partition has more than two elements. Define two new elements \linebreak[4] $(p^{Z=0},P_{\bof{XY}|\bof{UV}}^{Z=0})$ by
\begin{eqnarray}
\nonumber p^{Z=0}&:=&p^{z'_1}+\cdots+ p^{z'_m}\\
\nonumber P_{\bof{XY}|\bof{UV}}^{Z=0}&:=&\frac{1}{p^{z_1}}\sum\limits_{i=1}^m p^{z'_i}P_{\bof{XY}|\bof{UV}}^{z'_i},
\end{eqnarray}
where the set $z'_1,\ldots,z'_m$ is defined to consist of the boxes such that $P[\bigoplus_i X_i=0|Z=z'_i]> 1/2$. Similarly define 
$(p^{Z=1},P_{\bof{XY}|\bof{UV}}^{Z=1})$ as the convex combination of the remaining elements of the box partition. 
Because the spaces of boxes is convex, this forms again a valid box partition and it has the same distance. 
\pe
\end{proof}

It is, therefore, sufficient to consider a box partition with only two elements
$z=0$ and $z=1$. 
However, given one element of the box partition $(p,P_{\bof{XY}|\bof{UV}}^{Z=0})$, the second element 
$(1-p,P_{\bof{XY}|\bof{UV}}^{Z=1})$ is determined, 
because their convex combination forms the marginal box, $P_{\bof{XY}|\bof{UV}}$. 
\begin{lemma}\label{lemma:termsz0}
Assume a box partition $\bar{w}$ with element $(p,P_{\bof{XY}|\bof{UV}}^{Z=0})$ and 
an unbiased bit $S=f(\bof{X})$ such that w.l.o.g. 
{$P[S=0|Z=0,Q]\geq 1/2$}. 
Then the distance from uniform of $S$ given the box partition $\bar{w}$ and $Q=(\bof{U}=\bof{u},\bof{V}=\bof{v},F=f)$ is 
\begin{eqnarray}
\nonumber d(S|Z(\bar{w}),Q)&=&
2\cdot 
p\cdot (P[S=0|Z=0,Q]-1/2)\ .
\end{eqnarray}
\end{lemma}
\begin{proof}
 \begin{eqnarray}
\nonumber d(S|Z(\bar{w}),Q)&=&
p\cdot 
(P[S=0|Z=0,Q]-1/2
)
\\\nonumber &&+
(1-p)\cdot(-1) 
( \frac{1/2-p\cdot P[S=0|Z=0,Q]}{1-p}-1/2
)\\
\nonumber&=&
2\cdot 
p\cdot (P[S=0|Z=0,Q]-1/2)\ .
\end{eqnarray}
\peon
\end{proof}

The above lemmas imply that finding the distance from uniform is equivalent to finding the ``best'' element of a box partition 
$(p,P_{\bof{XY}|\bof{UV}}^{Z=0})$. When can $(p,P_{\bof{XY}|\bof{UV}}^{Z=0})$ be element of a box partition? 
The criterion is given in Lemma~\ref{lemma:zweihi}.
\begin{lemma}\label{lemma:zweihi}
Given a box $P_{\bof{XY}|\bof{UV}}$, there exists a box partition with element 
$(p,P_{\bof{XY}|\bof{UV}}^{Z=0})$ if and only if for all inputs and outputs $\bof{x},\bof{y},\bof{u},\bof{v}$,
\begin{eqnarray}
\label{zcanoccurwithp}  p\cdot P_{\bof{XY}|\bof{UV}}^{Z=0}(\bof{xy}|\bof{uv}) & \leq P_{\bof{XY}|\bof{UV}}(\bof{xy}|\bof{uv})\ .
\end{eqnarray}
\end{lemma}
\begin{proof}
The non-signaling condition is linear and the space of conditional probability distributions is convex, therefore a 
convex combination of  valid boxes $P_{\bof{XY}|\bof{UV}}^{Z=z}$ is again a valid box. To prove that the outcome $z=0$ can occur with probability $p$ it is, therefore, sufficient to show that there exists another valid  outcome $z=1$ which can occur with $1-p$, and that the weighted sum of the two is $P_{\bof{XY}|\bof{UV}}$. If $P_{\bof{XY}|\bof{UV}}^{Z=0}$ is a normalized and non-signaling probability distribution, then so is $P_{\bof{XY}|\bof{UV}}^{Z=1}$, because the sum of the two, $P_{\bof{XY}|\bof{UV}}$, is also non-signaling and normalized. Therefore, we only need to verify that all entries of the complementary box $P_{\bof{XY}|\bof{UV}}^{Z=1}$ are between $0$ and $1$.  However, this  box is the difference
\begin{eqnarray}
\nonumber P_{\bof{XY}|\bof{UV}}^{Z=1}=\frac{1}{1-p}(P_{\bof{XY}|\bof{UV}}-p\cdot P_{\bof{XY}|\bof{UV}}^{Z=0})\ .
\end{eqnarray}
Requesting this to be greater or equal to $0$ is equivalent to (\ref{zcanoccurwithp}). We observe that all entries of $P_{\bof{XY}|\bof{UV}}^{Z=1}$ are now trivially 
smaller than or equal to  $1$ because of the normalization: if the sum of positive summands is 1, each of them can be at most 1. \pe
\end{proof}

We can now show that the maximal distance from uniform which can be reached by a non-signaling adversary is the solution 
of a linear programming problem (see Appendix~\ref{sec:lin_prog} for details on linear programming).\footnote{In the following we drop 
the indices of the probability distributions as they should be clear from the context.} 
We introduce a new variable $\Delta$, which is a 
vector such that with each value of $\bof{x},\bof{y},\bof{u},\bof{v}$, we can associate an entry of $\Delta$ and we write 
$\Delta(\bof{x},\bof{y}|\bof{u},\bof{v})$ for this entry. $\Delta$ can be seen as a probability distribution describing a box, where the distribution need not be normalized nor positive. 
\begin{lemma}\label{lemma:distanceislp}
The distance from uniform of $\bigoplus_i X_i$ given $Z(W)$ and $Q:=(\bof{U}=\bof{u},\bof{V}=\bof{v},F=\bigoplus)$ is 
\begin{eqnarray}
\nonumber d(\bigoplus\nolimits_i X_i|Z(W),Q)&=& 1/2\cdot b^T\cdot \Delta^*\ ,
\end{eqnarray}
where $b^T\cdot \Delta^*$ is the optimal value of the linear program
\begin{eqnarray}
\label{eq:primal2} \text{max: }&& \sum_{(\bof{x},\bof{y}):f(\bof{x})=0} \Delta(\bof{xy}|\bof{uv})-\sum_{(\bof{x},\bof{y}):f(\bof{x})=1} \Delta(\bof{xy}|\bof{uv})\\
\nonumber \text{s.t.: }&& \sum_{\bof{x}}\Delta(\bof{xy}|\bof{uv})-\sum_{\bof{x}} \Delta(\bof{xy}|\bof{u'v})=0\ \forall \bof{y},\bof{v},\bof{u},\bof{u'}\ \text{ (non-signaling from Alice to Bob) }\\
\nonumber && \sum_{\bof{y}}\Delta(\bof{xy}|\bof{uv})-\sum_{\bof{y}} \Delta(\bof{xy}|\bof{uv'})=0\ \forall \bof{x},\bof{u},\bof{v},\bof{v'}\ \text{ (non-signaling from Bob to Alice) }\\
\nonumber && \Delta(\bof{xy}|\bof{uv}) \leq P(\bof{xy}|\bof{uv})\ \ \forall \bof{x},\bof{y},\bof{u},\bof{v}\\
\nonumber && \Delta(\bof{xy}|\bof{uv}) \geq -P(\bof{xy}|\bof{uv})\ \ \forall \bof{x},\bof{y},\bof{u},\bof{v}
\end{eqnarray}
\end{lemma}
\begin{proof}
We show that every element of a box partition $(p,P_{\bof{XY}|\bof{UV}}^{Z=0})$ corresponds to a feasible $\Delta$ and \emph{vice versa}.\\
Assume an element of a box partition $(p,P_{\bof{XY}|\bof{UV}}^{Z=0})$ and define
\begin{eqnarray}
\nonumber \Delta(\bof{xy}|\bof{uv})&=&2 p\cdot  P^{Z=0}(\bof{xy}|\bof{uv})-P(\bof{xy}|\bof{uv})\ .
\end{eqnarray}
$\Delta$ fulfills the non-signaling conditions by linearity. Further $p\geq 0$ and $P^{Z=0}(\bof{xy}|\bof{uv})\geq 0$ imply $\Delta(\bof{xy}|\bof{uv})\geq -P(\bof{xy}|\bof{uv})$ and $p\cdot P^{Z=0}(\bof{xy}|\bof{uv}) \leq P(\bof{xy}|\bof{uv})$ implies $\Delta(\bof{xy}|\bof{uv})\leq P(\bof{xy}|\bof{uv})$. $\Delta$ is, therefore, feasible. \\
To see the reverse direction, assume a feasible $\Delta$. Define
\begin{eqnarray}
\nonumber p&=&1/2\cdot (1+\sum_{\bof{xy}}\Delta(\bof{xy}|0\ldots 00\ldots 0))\\
\nonumber P^{Z=0}(\bof{xy}|\bof{uv})&=&\frac{P(\bof{xy}|\bof{uv})+\Delta(\bof{xy}|\bof{uv})}{2p}\ .
\end{eqnarray}
(For completeness, define $P^{Z=0}(\bof{xy}|\bof{uv})=P(\bof{xy}|\bof{uv})$ in case $p=0$.)
To see that $(p,P_{\bof{XY}|\bof{UV}}^{Z=0})$ is element of a box partition note that 
$\sum_{\bof{xy}}\Delta(\bof{xy}|0\ldots 00\ldots 0)=\sum_{\bof{xy}}\Delta(\bof{xy}|\bof{u'v'})$ for all $\bof{u'},\bof{v'}$ because of the non-signaling constraints. I.e., $p$ is independent of the chosen input and the transformation is, therefore, linear. This implies that $P^{Z=0}$ is still non-signaling. Because 
\begin{eqnarray}
\nonumber \sum_{\bof{xy}}P^{Z=0}(\bof{xy}|\bof{uv})=\sum_{\bof{xy}}\frac{P(\bof{xy}|\bof{uv})+\Delta(\bof{xy}|\bof{uv})}{2p}=\frac{1+(2p-1)}{2p}
=1
\end{eqnarray}
it is normalized. Because $-P(\bof{xy}|\bof{uv})\leq \Delta(\bof{xy}|\bof{uv})\leq P(\bof{xy}|\bof{uv})$ and $\sum_{\bof{xy}}P(\bof{xy}|\bof{uv})=1$, we have $-1\leq \sum_{\bof{xy}}\Delta(\bof{xy}|0\ldots 00\ldots 0)\leq 1$ and this implies $P^{Z=0}(\bof{xy}|\bof{uv})\geq 0$
i.e., $P^{Z=0}_{\bof{XY}|\bof{UV}}$ is a box. By Lemma~\ref{lemma:zweihi}, $(p,P_{\bof{XY}|\bof{UV}}^{Z=0})$ is element of a box partition because 
\begin{eqnarray}
\nonumber p\cdot P^{Z=0}(\bof{xy}|\bof{uv})&=& 1/2\cdot (1+\sum_{\bof{xy}}\Delta(\bof{xy}|0\ldots 00\ldots 0))\cdot \frac{P(\bof{xy}|\bof{uv})+\Delta(\bof{xy}|\bof{uv})}{1+\sum_{\bof{xy}}\Delta(\bof{xy}|0\ldots 00\ldots 0)}\\
\nonumber &=& 1/2\cdot (P(\bof{xy}|\bof{uv})+\Delta(\bof{xy}|\bof{uv}))\leq P(\bof{xy}|\bof{uv})\ .
\end{eqnarray}
Finally, we show that the value of the objective function for any $\Delta$ is exactly twice the distance from uniform reached by the box partition with element $(p,P^{Z=0}_{\bof{XY}|\bof{UV}})$:
\begin{eqnarray}
\nonumber &&\sum_{(\bof{x},\bof{y}):f(\bof{x})=0} \Delta(\bof{xy}|\bof{uv})-\sum_{(\bof{x},\bof{y}):f(\bof{x})=1} \Delta(\bof{xy}|\bof{uv})\\
\nonumber &=& 
\sum_{(\bof{x},\bof{y}):f(\bof{x})=0} \left(2p\cdot P^{Z=0}(\bof{xy}|\bof{uv})-P(\bof{xy}|\bof{uv})\right)-\sum_{(\bof{x},\bof{y}):f(\bof{x})=1}  \left(2p\cdot P^{Z=0}(\bof{xy}|\bof{uv})-P(\bof{xy}|\bof{uv})\right)\\
\nonumber &=& 
2p\cdot \left(\sum_{(\bof{x},\bof{y}):f(\bof{x})=0} 
P^{Z=0}(\bof{xy}|\bof{uv})-
\sum_{(\bof{x},\bof{y}):f(\bof{x})=1} 
P^{Z=0}(\bof{xy}|\bof{uv})\right)\\
\nonumber &=& 
2\cdot 2p(P[f(\bof{X})=0|Z=0,Q]-1/2)\ ,
\end{eqnarray}
which is exactly twice the distance from uniform by Lemma~\ref{lemma:termsz0}.
\pe
\end{proof}

We know that there exists a feasible $\Delta$ which reaches a value of 
$\prod_i(4\ep_i)$, namely the $\Delta$ associated with the box partition corresponding to an individual attack. We now want to show that this value is also dual feasible and, therefore, optimal. 
First, we re-write the primal in a form with only inequality constraints and no equality constraints. To do so, we replace constraints of the form $a_j\cdot \Delta=0$ by the two constraints 
$a_j\cdot \Delta\leq 0$ and $-a_j\cdot \Delta\leq 0$. 
We obtain:
\begin{eqnarray}
\label{eq:primal_delta}
\begin{array}{lcr}
\text{max: }&& b^T\cdot \Delta\\
\text{s.t.: }&& 
A
\cdot \Delta \leq
c\\
\phantom{s}
\end{array}
&\text{\ \ \ \ \ \ \ \ \ \ \ \ \ \ and its dual\ \ \ \ \ \ \ \ \ \ \ \ \ \ }&
\begin{array}{lcr}
\text{min: }&& c^T \lambda\\
 \text{s.t.: }&& A^T\cdot \lambda=b\\
 && \lambda\geq 0
\end{array}
\end{eqnarray}
The explicit values of $A,b,c$ and the dual optimal solution $\lambda^*$ for the case of a single box are given in Appendix~\ref{sec:explicite_1box}. Note that in the \emph{dual} program, the marginal box as seen by Alice and Bob only appears in the objective function. The feasible region is, therefore, completely independent of the marginal. 

Our main tool to show optimality will be to show that we can express the linear program describing $n$ boxes as the tensor product of the linear program describing one box. 
\begin{lemma}\label{lemma:product_form}
Assume $A_1,b_1,c_1$ are the vectors and matrices associated with the linear program (\ref{eq:primal_delta}) for the case of a single box. Then the value of the program $A,b,c$ associated with $n$ boxes is equal to the value of the linear program 
defined by\footnote{We write here $c_1$ for each of the $n$ boxes for notational simplicity. However, the marginal box $c_i$ could actually 
be different for each of the $n$ boxes without having to change our argument. 
} 
\begin{eqnarray}
\label{eq:product_lp}
\text{max: }&& (b_1^{\otimes n})^T\cdot \Delta\\
\nonumber \text{s.t.: }&& 
A_1^{\otimes n}
\cdot \Delta \leq
c_1^{\otimes n}\ .
\end{eqnarray}
\end{lemma}
\begin{proof}
We describe the case $n=2$, the case of larger $n$ is analogue.
First note that with each entry of $\Delta$ for a single box there are associated input and output bits $X_i,Y_i,U_i,V_i$. With each 
entry of $\Delta$ living in the tensor product space of two boxes, we can associate an entry $\bof{X},\bof{Y},\bof{U},\bof{V}$ 
corresponding to two bits each in the obvious way. \\
$b_1$ is such that the entries associated with $X_1=0,U_1=0,V_1=0$ is $1$; $X_1=1,U_1=0,V_1=0$ is $-1$ and for all other inputs 
it is zero (the choice of input $0,0$ is arbitrary and no restriction). $b_1\otimes b_1$ is, therefore, such 
that for $\bigoplus_i X_i=0,\bof{U}=00,\bof{V}=00$ it is $1$; for $\bigoplus_i X_i=1,\bof{U}=00,\bof{V}=00$ it is $-1$ and for all 
other inputs it is $0$. This is exactly the form that gives us the bias of the XOR of two output bits given input 
$\bof{U},\bof{V}=00,00$. \\
Now let us see that $A$ and $c$ can also be taken of tensor product form. Indeed, we will show that the constraints 
given by $A_1^{\otimes n}$ and $c_1^{\otimes n}$ are either exactly the ones that describe a $2n$ non-signaling box or 
they are trivially fulfilled and, therefore, do not modify the value of the linear program. 
We can divide the lines of $A$ into $4$ types, we call them $A^{n-s}$ (for ``non-signaling''), $-A^{n-s}$ (which contains the same 
coefficients as $A^{n-s}$ but with the sign reversed), $1_{16\times 16}$ and $-1_{16\times 16}$ (which contains a $1$ resp. $-1$ at 
a certain position and $0$ everywhere else) (compare with Appendix~\ref{sec:explicite_1box}). 
The entries of $c$ associated with these types are respectively $0$, $0$, $P(\bof{xy}|\bof{uv})$ and $P(\bof{xy}|\bof{uv})$ 
(the marginal probabilities). \\
Now consider $A_1^{\otimes 2}$ and $c_1^{\otimes 2}$. We now have $16$ types of rows, 
corresponding to all possible combinations. 
\begin{enumerate}
\item \label{item:id_id} 
Type $1_{16\times 16} \otimes 1_{16\times 16}=1_{256\times 256}$ 
The associated $c$ is $P(x_1y_1|u_1v_1) \cdot P(x_2y_2|u_2v_2)$ (i.e., the probability entry of the two boxes) 
and these constraints correspond exactly to the upper bound on $\Delta$ in the case of two boxes. 
(Type $-1_{16\times 16} \otimes -1_{16\times 16}=1_{256\times 256}$ is exactly the same row and, therefore, follows from this one).
\item \label{item:mid_id} 
Type $-1_{16\times 16} \otimes 1_{16\times 16}=-1_{256\times 256}$  
The associated $c$ is $P(x_1y_1|u_1v_1) \cdot P(x_2y_2|u_2v_2)$ and these constraints correspond exactly to the lower bound 
on $\Delta$ in the case of two boxes. (Type $1_{16\times 16} \otimes -1_{16\times 16}=-1_{256\times 256}$ is exactly 
the same row and, therefore, follows from this one).
\item \label{item:ns_id} The lines of the form $A^{n-s}\otimes 1_{16\times 16}$ correspond exactly to the non-signaling constraints 
for two boxes. To see this assume that the non-signaling constraint on the first box is of the form 
$\sum_{x_1} P(x_1,y_1|u_1,v_1)-\sum_{x_1}P(x_1,y_1|u'_1v_1)$ and the identity on the second box is 
$1$ at the position $x_2,y_2,v_2,u_2$ and $0$ everywhere else. Then the constraint $A^{n-s}\otimes 1_{16\times 16}$ corresponds to 
\begin{eqnarray}
 \nonumber \sum_{x_1} P(x_1,x_2,y_1,y_2|u_1,u_2,v_1,v_2)-\sum_{x_1}P(x_1,x_2,y_1,y_2|u'_1,u_2,v_1,v_2)
\end{eqnarray}
and this is exactly the form of a $2n$ non-signaling constraint. The associated entry of $c$ is, as expected, 
$0\cdot P(x_2,y_2|u_2,v_2)=0$. Together with the constraints of the form $1_{16\times 16}\otimes A^{n-s}$ we obtain all the non-signaling 
constraints for the two boxes. (Type $-A^{n-s}\otimes -1_{16\times 16}$ is again exactly the same row).
\item \label{item:mns_id} The lines of the form $-A^{n-s}\otimes 1_{16\times 16}$ and 
$1_{16\times 16}\otimes -A^{n-s}$ give the same non-signaling constraints as above but with reversed sign, therefore, enforcing the 
equality constraint by two inequality constraints. ($A^{n-s}\otimes -1_{16\times 16}$ and $-1_{16\times 16}\otimes A^{n-s}$ are again 
exactly the same rows and are, therefore, trivially fulfilled.)
\item Remain the lines of the form $A^{n-s}\otimes A^{n-s}$. Their associated $c$ is $0\cdot 0=0$. However, the second non-signaling 
constraints can be seen as a linear combination of the identity constraints, i.e., $A^{n-s}\otimes A^{n-s}=A^{n-s}\otimes 
(\sum_k \alpha_k\cdot 1_{16\times 16,k})$. Because of the linearity of the tensor product in the second component, this constraint is, 
therefore, the linear combination of the constraints given in point \ref{item:mid_id} and \ref{item:ns_id} above 
and because each of them is equal to $0$, their linear combination is also equal to $0$ and this constraint is, 
therefore, trivially fulfilled whenever the above constraints are. The same argument holds for the rows 
  $-A^{n-s}\otimes A^{n-s}$, $A^{n-s}\otimes -A^{n-s}$ and $-A^{n-s}\otimes -A^{n-s}$. 
\end{enumerate}
\peon
\end{proof}
Now we consider the dual program of (\ref{eq:product_lp}). Using Lemma~\ref{lemma:product_form} we see that if $\lambda_1$ is a feasible dual solution for a single box, then $\lambda_1^{\otimes n}$ is feasible for $n$ boxes. 
\begin{lemma}\label{lemma:dual_product}
For any $\lambda_i$ which is dual feasible for the linear program $A_1,b_1$ associated with one box, $\bigotimes_i \lambda_i$ is dual feasible for the linear program (\ref{eq:product_lp}) associated with $n$ boxes. Further, this dual feasible solution has value $c_n^T\lambda_n=\prod_i(c_i^T\lambda_i)$. 
\end{lemma}
\begin{proof}
$\lambda_i$ is dual feasible for $A_1,b_1$, i.e., $A_1^T\lambda_1=b_1$ and $\lambda_1\geq 0$. Then 
\begin{eqnarray}
\nonumber A_n^T\lambda_n=(A_1^{\otimes n})^T(\bigotimes_i\lambda_i)=(A_1^T)^{\otimes n}(\bigotimes_i\lambda_i)=\bigotimes_i(A_1^T\lambda_i)=(b_1)^{\otimes n}
\end{eqnarray}
and $\bigotimes_i(\lambda_i)\geq 0$, i.e., $\lambda_n=\bigotimes_i\lambda_i$ is dual feasible. Its value is $c_n\lambda_n=\bigotimes_ic_i\cdot \bigotimes_i\lambda_i=\bigotimes_i(c_i\lambda_i)=\prod_i(c_i\lambda_i)$.
\pe
\end{proof}

Now we are ready to give the proof of Lemma~\ref{lemma:xor_good_pa_fct}.
\begin{proof}[of Lemma~\ref{lemma:xor_good_pa_fct}]
For a single box $d(X|Z(W),Q)\leq 1/2\cdot (4\ep_i)$ by Lemma~\ref{lemma:guessing_probability_single_box} (see Section~\ref{sec:limit_one_box}), this implies that there exists a dual feasible $\lambda_i$, such that $c_i^T\lambda_i\leq 4\ep_i$ for each $i$. By Lemma~\ref{lemma:dual_product}, there exists a dual feasible $\lambda_n$ such that $c_n^T\lambda_n\leq \prod_i(4\ep_i)\leq (4\ep)^n$ 
and, therefore, by Lemma~\ref{lemma:distanceislp},
\begin{eqnarray}
\nonumber d(\bigoplus\nolimits_i X_i|Z(W),Q)=1/2\cdot c_n^T\lambda_n^*\leq 1/2\cdot  c_n^T\lambda_n=1/2\cdot \prod\nolimits_i (4\ep_i)\leq 1/2\cdot (4\ep)^n \ .
\end{eqnarray}
\peon
\end{proof}

This implies that if Alice and Bob create a single key bit by applying the XOR to their outputs there is no advantage for Eve to do a collective or coherent attack, as the above distance from uniform can be reached by an 
individual attack.\footnote{Note that individual attacks are optimal only in this specific case and, in general, they are strictly weaker than 
collective or coherent attacks. We give an example of such a collective attack in Appendix~\ref{sec:better_collective_attack}.}   

We now want to remove the condition that the marginal boxes of Alice and Bob need to be independent. First we consider the case when Alice and Bob share the \emph{convex combination} of $n$ independent and unbiased boxes of different errors. The reason to consider this case is because no matter what boxes Alice and Bob share --- they can be arbitrarily correlated --- Alice and Bob can apply a random mapping to their input and output bits (see Appendix~\ref{sec:depol}), such that the distribution they share after this mapping in fact \emph{is} the one of a convex combination of several independent and unbiased boxes with different errors~\cite{mag,mrwbc}. The statement of Lemma~\ref{lemma:xor_good_pa_fct} still holds here:
\begin{lemma}\label{lemma:xor_good_pa_fct_convex}
Assume a $(2n+1)$-partite box $P_{\bof{XY}Z|\bof{UV}W}$ such that the marginal $P_{\bof{XY}|\bof{UV}}$ corresponds to a convex combination with weight $p_j$ of $n$ unbiased non-local boxes each with an associated error $\ep_i^j$. Assume $f(\bof{X}):=\bigoplus_i X_i$ and $Q:=(\bof{U}=\bof{u},\bof{V}=\bof{v},F=\bigoplus)$. Then
$
d(f(\bof{X})|Z(W),Q)\leq \sum\nolimits_j p_j\cdot \left[ 1/2\cdot \prod\nolimits_i(4\ep_i^j)\right] 
$.
\end{lemma}
\begin{proof}
Note that for a single box the dual optimal solution is $\lambda_1^*$ for all $c_1$ describing a single box (i.e., $c_1^T\cdot \lambda_1^*=4\ep$ for all $c_1$) (see Appendix~\ref{sec:explicite_1box}). For $n$ boxes, $\lambda_1^{\otimes n}$ is still dual feasible. It reaches a value of $c_n^T\lambda_1^{\otimes n}=
\left(\sum_j p_j (\otimes_i c_i^j)\right)\cdot (\lambda_1)^{\otimes n}=\sum_j p_j \prod_i (4\ep_i^j)$. \pe
\end{proof}

Now we want to remove any requirement of independence. Lemma~\ref{lemma:nodepol} states  that choosing boxes which 
are not independent cannot be an advantage for Eve and the above bounds still hold. 
\begin{lemma}\label{lemma:nodepol}
Assume a $(2n+1)$-partite box $P_{\bof{XY}Z|\bof{UV}W}$ with any marginal $P_{\bof{XY}|\bof{UV}}$. Assume $f(\bof{X}):=\bigoplus_i X_i$ and $Q:=(\bof{U}=\bof{u},\bof{V}=\bof{v},F=\bigoplus)$ and the distance from uniform $d(f(\bof{X})|Z(W),Q)$. Now assume a second $(2n+1)$-partite box with  marginal $P'_{\bof{XY}|\bof{UV}}$ obtained from $P_{\bof{XY}|\bof{UV}}$ by depolarization and with distance from uniform $d'(f(\bof{X})|Z(W),Q)$ (with the same $Q$ and $f$). Then
$
d(f(\bof{X})|Z(W),Q)\leq d'(f(\bof{X})|Z(W),Q) 
$.
\end{lemma}
\begin{proof}
We know that for $P'_{\bof{XY}|\bof{UV}}$, $d'(f(\bof{X})|Z(W),Q)= \sum_j p_j\cdot 1/2\cdot \prod_i(4\ep_i^j)$ by Lemma~\ref{lemma:xor_good_pa_fct_convex} and because this bound can 
easily be reached by attacking each box separately. 
However, this value is exactly the 
sum of all probabilities where \emph{none} of the CHSH conditions are fulfilled (i.e., where $X_i\oplus Y_i\neq U_i\cdot V_i$ for all $i$). \\
Now consider $P_{\bof{XY}|\bof{UV}}$. $P'_{\bof{XY}|\bof{UV}}$ can be seen as the convex combination of all the $P_{\bof{XY}|\bof{UV}}$ to which one of the mappings given in Appendix~\ref{sec:depol} has been applied. However, the distance from uniform for $P_{\bof{XY}|\bof{UV}}$ (or their mappings) is limited by the sum of all probabilities where \emph{none} of the CHSH conditions are fulfilled and this holds for all values of the input $\bof{u},\bof{v}$ (in Appendix~\ref{sec:explicite_1box}, for each input $\bof{u},\bof{v}$ a dual feasible solution reaching this value is given). By comparison with Appendix~\ref{sec:depol}, we see that the mappings (for each input) leave $\bigoplus_i x_i$ unchanged (up to a relabeling between $0$ and $1$). The mappings also leave the sum of probabilities where \emph{none} of the CHSH conditions are fulfilled unchanged, because $x'_i,y'_i,u'_i,v'_i$ not fulfilling the CHSH condition are mapped to $x_i,y_i,u_i,v_i$ not fulfilling the CHSH condition. We conclude
\begin{eqnarray}
 \nonumber d(f(\bof{X})|Z(W),Q)&\leq& \lambda_1^{*\ \otimes n}\cdot c =  \sum_{\bof{x},\bof{y},\bof{u},\bof{v}:x_i\oplus y_i 
 \neq u_i\cdot v_i\ \forall i } P_{\bof{XY}|\bof{UV}}(\bof{x}\bof{y}|\bof{u}\bof{v})
\\
\nonumber &=&
\sum_{\bof{x},\bof{y},\bof{u},\bof{v}:x_i\oplus y_i \neq u_i\cdot v_i\ \forall i } 
P'_{\bof{XY}|\bof{UV}}(\bof{x}\bof{y}|\bof{u}\bof{v})=
 d'(f(\bof{X})|Z(W),Q)\ .
\end{eqnarray}
 \peon
\end{proof}

\section{Full Key Agreement}\label{sec:keygeneration}

\subsection{Privacy Amplification: From One to Several Bits}

We have seen in the previous section that it is possible to create a highly secure bit using a linear function --- 
the XOR. But obviously we would like to extract a secure key instead of a single bit. Alice and Bob will create all the 
key bits the same way: by applying a random linear function to the output bits, i.e., $S:=A\odot \bof{X}$, where $A$ is a 
$s\times n$-matrix over $GF(2)$ with $p(0)=p(1)=1/2$ for all entries and we write $\odot$ for the multiplication modulo 
$2$. Let us now see why this key is secure. 

First, we reduce the security of the key $S$ to the question of the security of every single bit.
\begin{lemma}\label{lemma:distanceseveralbits}
 Assume $S:=[S_1,\ldots ,S_s]$, where $S_i$ are bits. Then
\begin{eqnarray}
d(S|Z(W),Q)\leq 
\sum_i d(S_i|Z(W),Q,S_{1},\ldots,S_{i-1})\ .
\end{eqnarray}
\end{lemma}
\begin{proof}
\begin{eqnarray}
\nonumber d(S|Z(W),Q)&=& \sum_{s,q}\max_w \sum_z |P_{S,Z,Q|W=w}(s,z,q)-\frac{1}{2^s}\cdot P_{Z,Q|W=w}(z,q)|\\
 \nonumber &\leq & \sum_{s,q}\max_w \sum_z \left[|P_{S,Z,Q|W=w}(s,z,q)-\frac{1}{2}\cdot P_{S_1\ldots S_{s-1},Z,Q|W=w}(s_1,\ldots ,s_{s_1},z,q)| \right.\\
\nonumber && 
\left. +\ldots +\frac{1}{2^{s-1}}|P_{S_1,Z,Q|W=w}(s_1,z,q)-\frac{1}{2}\cdot P_{Z,Q|W=w}(z,q)|
 \right]\\
\nonumber 
&\leq&\sum_i d(S_i|Z(W),Q,S_{1},\ldots ,S_{i-1})\ ,
\end{eqnarray}
where the first equation is by the definition of the distance from uniform and the second inequality is by the triangle inequality. \pe
\end{proof}

We now need to bound the distance from uniform of the $i$'th key bit given all previous bits. 
\begin{lemma}\label{lemma:distance_k_th_bit}
Assume $S:=A\odot \bof{X}$, where $A$ is a $i\times n$-matrix over $GF(2)$ and be $P_A$ the uniform distribution 
over all these matrices. $Q:=(\bof{U}=\bof{u},\bof{V}=\bof{v},A)$. Then
\begin{eqnarray}
\label{eq:distance_k_th_bit}
d(S_i|Z(W),Q,S_1,\ldots,S_{i-1})
&\leq& 1/2\cdot 
2^{i-1}
\cdot  \left(\frac{1+4\ep}{2}\right)^n.
\end{eqnarray}
\end{lemma}
\begin{proof}
Bounding the distance from uniform of $S_i$ given $S_1,\ldots,S_{i-1}$ corresponds to bounding the distance from uniform of 
$S_i$ given all linear combinations over $GF(2)$ of $S_1,\ldots,S_{i-1}$ (see Appendix~\ref{sec:distance_set_given_other_sets}). 
For each linear combination $\bigoplus_{j\in I}S_j$ define the random bit $S_c=c\odot \bof{X}$ where 
$c=\bigoplus_{j\in I}a_j\oplus a_i$ and $a_j$ denotes the $j$�th line of the matrix $A$. Note that $S_c$ is a random linear function 
over $\bof{X}$ (the proof of this is given in Appendix~\ref{sec:lin_com_of_random_vect}). If $S_c$ is uniform and independent of 
$S_1,\ldots,S_{i-1}$, then $S_i$ is uniform given this specific linear combination. However, the distance from uniform and independent of 
$S_c$ is given by Lemma~\ref{lemma:xor_good_pa_fct} (note that Lemma~\ref{lemma:xor_good_pa_fct}  bounds not only the 
distance from uniform of $S_c$ given $Z$, but also given all $X_i$ not included in $S_c$, as these could be included in the variable $Z$). 
We obtain  
\begin{eqnarray}
\nonumber d(c\odot \bof{X} |Z(W),Q)
&\leq & 1/2\cdot \frac{1}{2^n}\sum_{K\subseteq n}\prod_{i\in K}(4\ep_i)
\leq 1/2\cdot  \frac{1}{2^n}\sum_{k=0}^{n}\binom{n}{k}(4\ep)^k
=1/2\cdot \left(\frac{1+4\ep}{2}\right)^n\ ,
\end{eqnarray}
where the second inequality follows from the fact that this expression is 
maximized when all $\ep_i$ are equal (see Appendix~\ref{subsec:ep_average} for a proof of this). If a random variable $S$ has 
distance from uniform at most $d$, then we can define an event $\mathcal{E}$ occuring with probability 
at least $1-d$ such that given $\mathcal{E}$, $S$ is uniform. By the union bound over all 
$2^{i-1}$ 
possible linear combinations of 
$S_1,\ldots,S_{i-1}$, we obtain the probability that $S_i$ is uniform given $S_1,\ldots,S_{i-1}$ and, therefore, the bound on the distance
from uniform 
\begin{eqnarray}
d(S_i|Z(W),Q,S_1,\ldots,S_{i-1})
&\leq& 1/2\cdot 2^i\cdot  \left(\frac{1+4\ep}{2}\right)^n.
\end{eqnarray}
\peon
\end{proof}

Now we can bound the distance from uniform of a key $S:=S_1\ldots S_s$ by Lemma~\ref{lemma:distanceseveralbits} 
and~\ref{lemma:distance_k_th_bit}. 
\begin{lemma}\label{lemma:distance_key_string}
Assume $S:=A\odot \bof{X}$, where $A$ is a $s\times n$-matrix over $GF(2)$ and be 
$P_A$ the uniform distribution over all these matrices. $Q:=(\bof{U}=\bof{u},\bof{V}=\bof{v},A)$. Then
\begin{eqnarray}
d(S|Z(W),Q)
&\leq& 1/2\cdot 
2^{s}
\cdot  \left(\frac{1+4\ep}{2}\right)^n.
\end{eqnarray}
\end{lemma}
\begin{proof}
By Lemma~\ref{lemma:distanceseveralbits} and~\ref{lemma:distance_k_th_bit}
\begin{eqnarray}
\nonumber 
d(S|Z(W),Q)&\leq& 1/2\cdot  \left(\frac{1+4\ep}{2}\right)^n \cdot 
(\sum_{i=1}^s2^{i-1})
\leq 1/2\cdot  \left(\frac{1+4\ep}{2}\right)^n \cdot 
(\frac{2^s-1}{2-1})
\leq 1/2\cdot 
2^{s}
\cdot  \left(\frac{1+4\ep}{2}\right)^n,
\end{eqnarray}
where the second inequality follows from the expression for geometric series. \pe
\end{proof}

\subsection{Information Reconciliation}

In general, the outputs $\bof{x}$ and $\bof{y}$ of Alice and Bob are not equal but have a certain probability to differ. 
Alice and Bob, therefore, need to do information reconciliation. They can do this the same way they create the key, 
namely by using a random linear code. This follows directly from a result from~\cite{carterwegman} about two-universal sets of 
hash functions and from a result from~\cite{brassardsalvail} about information reconciliation. We restate the theorems below. 
\begin{theorem}[\cite{carterwegman}]
The set of functions $f_A(\bof{x}):=A\odot \bof{x}$, where $A$ is any $n\times m$-matrix over $GF(2)$ is two-universal.
\end{theorem}

\begin{theorem}[\cite{brassardsalvail}]\label{th:ir}
 Suppose an $n$-bit string $\bof{x}$ another $n$-bit string $\bof{y}$ obtained by sending $\bof{x}$ over a binary symmetric 
 channel with error parameter $\delta$. Assume the 
 function $f:\{0,1\}^n\rightarrow \{0,1\}^m$ is chosen at random amongst a set of two-universal functions.  
 Choose $\bof{y'}$ such that $d_H(\bof{y},\bof{y'})$ is minimal among all strings $\bof{r}$ with $f(\bof{r})=f(\bof{x})$. 
 Then $P_{\bof{x}\neq \bof{y'}}\leq 1-e^{-2^{n\cdot h(\delta+\epsilon)-m}}+\frac{(\log{n})^2\ep(1-\delta)}{n}$.
\end{theorem}
The above theorems show that in the limit of large $n$, $m=\lceil n\cdot h(\delta)\rceil$ (where $\delta$ is the probability that Bob's 
bit is different from Alice's and $h$ the binary entropy function),  is both necessary and sufficient for Bob to correct the 
errors in his raw key, i.e., 
the protocol is $\epsilon'$-correct for any $\epsilon'>0$.

If Alice and Bob communicate $m$ bits during the information reconciliation phase, then the security of the key after information 
reconciliation can be calculated by replacing in Lemma~\ref{lemma:distance_key_string} the length of the key by the length of the key 
plus information reconciliation, i.e., $s\mapsto s+m$ and we obtain the following lemma.
\begin{lemma}\label{lemma:distance_key_with_ir}
Assume $[S,R]:=A\odot \bof{X}$,where $A$ is a $(s+m)\times n$-matrix over $GF(2)$ and be $P_A$ the uniform distribution 
over all these matrices. $Q:=(\bof{U}=\bof{u},\bof{V}=\bof{v},A)$. Then
\begin{eqnarray}
 d(S|Z(W),Q,R)&\leq&  1/2\cdot 
2^{s+m}
 \cdot  \left(\frac{1+4\ep}{2}\right)^n.
\end{eqnarray}
\end{lemma}

\subsection{Key Rate}

The key rate is the length of the key divided by the number of boxes used in the limit of a large number of boxes. Because we only need a 
small number of boxes for parameter estimation (see Appendix~\ref{sec:parameter_estimation}), this will asymptotically 
correspond to $q:=s/n$. 
From Lemma~\ref{lemma:distance_key_with_ir}, we can calculate the key rate by setting $m:=h(\delta)\cdot n$ 
(see also Protocol~\ref{prot} in Section~\ref{sec:protocol} for a detailed description of the protocol).
\begin{lemma}\label{lemma:keyrate}
The protocol reaches a key rate $q$ of
\begin{eqnarray}
q=1- h(\delta)-\log_2 (1+4\ep).
\end{eqnarray}
\end{lemma}
\begin{proof}
 From Lemma~\ref{lemma:distance_key_with_ir} and by the definition of the key rate, we can see that the protocol reaches a key 
 rate $q$ if
\begin{eqnarray}
\nonumber 2^{h(\delta)}\cdot 2^q \cdot \frac{1+4\ep}{2}<1. 
\end{eqnarray}
\peon
\end{proof}

Corollary~\ref{corollary:positive_key_rate} states for which parameters key agreement is possible (see Figure~\ref{fig:key_quantum_area}).
\begin{corollary}\label{corollary:positive_key_rate}
 The protocol reaches a positive key rate if $\ep<2^{-h(\delta)-1}-1/4$.
\end{corollary}

If the boxes have the same error for all inputs ($\delta=\ep$) then $m:=n\cdot h(\ep)$ and the protocol does not reach a 
positive secret key rate for $\ep=\frac{1+\sqrt{2}}{4}$, the minimum value reachable by quantum mechanics. To reach a positive key 
rate using quantum mechanics, Alice and Bob will, therefore, need to use different boxes, as described in the next section. 

\subsection{The Quantum Regime}

To get a positive key rate in the quantum regime, Alice and Bob use a box which gives highly correlated output bits given input $(0,0)$ (see 
Figure~\ref{eq:box_for_key_gen}) and generate their raw key only from these outputs.
\footnote{Another way to reach a positive 
key rate in the quantum regime is to use a type of non-locality characterized by a different Bell inequality allowing for a 
higher violation in the quantum regime. See~\cite{lluis} for details.} 
The parameter limiting 
Eve's knowledge  is then still $\ep=1/4\cdot \sum_{x\oplus y\neq u\cdot v}P_{XY|UV}(x,y,u,v)$, 
the parameter defining the amount of information reconciliation necessary is, however, the error in the correlation given 
input $(0,0)$ ($\delta$ in Figure~\ref{eq:box_for_key_gen}). 
Note that in a noiseless setting the distribution described in black font can be achieved by measuring a singlet 
state (see Protocol~\ref{prot} below). In that case, Alice and Bob will have perfectly correlated bits (and therefore would not need 
to do 
any information reconciliation), and the parameter limiting Eve's knowledge is $\ep=0.1875$.  
The parameters $\delta$ and $\epsilon$ (in light gray font in Figure~\ref{eq:box_for_key_gen}) are introduced to 
account for the noise in the state and/or measurement. 
\begin{figure}[ht]
\centering
\psset{unit=0.525cm}
\pspicture*[](-2,-1)(8.5,10)
\psline[linewidth=0.5pt]{-}(0,6)(-1,7)
\rput[c]{0}(-0.25,6.75){\normalsize{$X$}}
\rput[c]{0}(-0.75,6.25){\normalsize{$Y$}}
\rput[c]{0}(-0.5,7.5){\Large{$U$}}
\rput[c]{0}(-1.5,6.5){\Large{$V$}}
\rput[c]{0}(2,7.5){\Large{$0$}}
\rput[c]{0}(6,7.5){\Large{$1$}}
\rput[c]{0}(1,6.5){\Large{$0$}}
\rput[c]{0}(3,6.5){\Large{$1$}}
\rput[c]{0}(5,6.5){\Large{$0$}}
\rput[c]{0}(7,6.5){\Large{$1$}}
\rput[c]{0}(-1.5,4.5){\Large{$0$}}
\rput[c]{0}(-1.5,1.5){\Large{$1$}}
\rput[c]{0}(-0.5,5.25){\Large{$0$}}
\rput[c]{0}(-0.5,3.75){\Large{$1$}}
\rput[c]{0}(-0.5,2.25){\Large{$0$}}
\rput[c]{0}(-0.5,0.75){\Large{$1$}}
\psline[linewidth=2pt]{-}(-1,0)(8,0)
\psline[linewidth=2pt]{-}(-1,6)(8,6)
\psline[linewidth=2pt]{-}(-1,3)(8,3)
\psline[linewidth=1pt]{-}(0,1.5)(8,1.5)
\psline[linewidth=1pt]{-}(0,4.5)(8,4.5)
\psline[linewidth=2pt]{-}(0,0)(0,7)
\psline[linewidth=2pt]{-}(8,0)(8,7)
\psline[linewidth=2pt]{-}(4,0)(4,7)
\psline[linewidth=1pt]{-}(2,0)(2,6)
\psline[linewidth=1pt]{-}(6,0)(6,6)
\rput[c]{0}(1,5.25){\Large{$\frac{1}{2}\textcolor{gray}{-\frac{\delta}{2}} $}}
\rput[c]{0}(3,3.75){\Large{$\frac{1}{2}\textcolor{gray}{-\frac{\delta}{2}} $}}
\rput[c]{0}(5,5.25){\Large{$\frac{3}{8}\textcolor{gray}{-\frac{\epsilon}{2}} $}}
\rput[c]{0}(7,3.75){\Large{$\frac{3}{8}\textcolor{gray}{-\frac{\epsilon}{2}} $}}
\rput[c]{0}(1,2.25){\Large{$\frac{3}{8}\textcolor{gray}{-\frac{\epsilon}{2}} $}}
\rput[c]{0}(3,0.75){\Large{$\frac{3}{8}\textcolor{gray}{-\frac{\epsilon}{2}} $}}
\rput[c]{0}(5,0.75){\Large{$\frac{3}{8}\textcolor{gray}{-\frac{\epsilon}{2}} $}}
\rput[c]{0}(7,2.25){\Large{$\frac{3}{8}\textcolor{gray}{-\frac{\epsilon}{2}} $}}
\rput[c]{0}(3,5.25){\Large{$\textcolor{gray}{\frac{\delta}{2}} $}}
\rput[c]{0}(1,3.75){\Large{$\textcolor{gray}{\frac{\delta}{2}} $}}
\rput[c]{0}(7,5.25){\Large{$\frac{1}{8}\textcolor{gray}{+\frac{\epsilon}{2}} $}}
\rput[c]{0}(5,3.75){\Large{$\frac{1}{8}\textcolor{gray}{+\frac{\epsilon}{2}} $}}
\rput[c]{0}(3,2.25){\Large{$\frac{1}{8}\textcolor{gray}{+\frac{\epsilon}{2}} $}}
\rput[c]{0}(1,0.75){\Large{$\frac{1}{8}\textcolor{gray}{+\frac{\epsilon}{2}} $}}
\rput[c]{0}(5,2.25){\Large{$\frac{1}{8}\textcolor{gray}{+\frac{\epsilon}{2}} $}}
\rput[c]{0}(7,0.75){\Large{$\frac{1}{8}\textcolor{gray}{+\frac{\epsilon}{2}} $}}
\endpspicture
\caption{The quantum box used for key agreement.}
\label{eq:box_for_key_gen}
\end{figure}

\subsection{The Protocol}\label{sec:protocol}

In the following we give a detailed description of  our key agreement protocol. 
\begin{protocol}\label{prot}\ 
\begin{enumerate}
 \item Alice creates $n+k$ maximally entangled states $\ket{\Psi^-}=\frac{1}{\sqrt{2}}(\ket{01}-\ket{10})$, for some $k=\Theta(n)$, and sends one Qbit of every state to Bob.
 \item Alice and Bob randomly measure the $i$'th system in either the basis $U_0$ or $U_1$ (for Alice) or $V_0$ and $V_1$ (Bob); the four bases are shown
 in Figure~\ref{basen}.
All the $2(n+k)$ measurement events are {\em pairwise space-like separated\/}. 
 \item They randomly choose
$n$ of the measurement results when both measured $U_0,V_0$ to form the raw key.
\item For the remaining $k$ measurements they announce the results over the public authenticated channel and estimate the parameters $\ep$ and $\delta$ (see Appendix~\ref{sec:parameter_estimation}). They also check whether they have obtained roughly the same number of $1$'s and $0$'s (for the information reconciliation scheme).
If the parameters are such that key agreement is possible (Figure~\ref{fig:key_quantum_area}) they continue; otherwise they abort.
\item Information reconciliation and privacy amplification: Alice randomly chooses a $(m+s)\times n$-matrix $A$ such that $p(0)=p(1)=1/2$ for all entries and $m:=\lceil n\cdot h(\delta)\rceil$. She calculates
 $A\odot \bof{x}$ (where $\bof{x}$ is Alice's raw key) and sends the first $m$ bits to Bob over the public authenticated 
 channel. The remaining bits form the key. 
\end{enumerate}
\end{protocol}

\begin{figure}[ht]
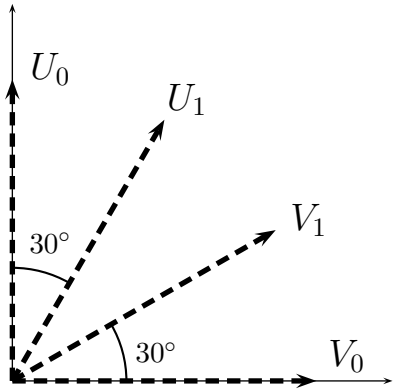

\centering
\pspicture*[](-0.5,-0.5)(5.5,5.5)
\psarc(0,0){1.5}{60}{90}
\rput[B]{0}(0.5,1.7){\normalsize{$30^{\circ}$}}
\psarc(0,0){1.5}{0}{30}
\rput[B]{0}(1.9,0.25){\normalsize{$30^{\circ}$}}
\rput[B]{0}(0.5,4){\Large{$U_0$}}
\rput[br]{90}(0,0){\psline[linewidth=0.5pt]{->}(0,0)(5,0)}
\rput[br]{0}(0,0){\psline[linewidth=0.5pt]{->}(0,0)(5,0)}
\rput[B]{0}(2.3,3.6){\Large{$U_1$}}
\rput[br]{90}(0,0){\psline[linewidth=2pt,linestyle=dashed]{->}(0,0)(4,0)}
\rput[br]{60}(0,0){\psline[linewidth=2pt,linestyle=dashed]{->}(0,0)(4,0)}
\rput[B]{0}(4.4,0.2){\Large{$V_0$}}
\rput[br]{30}(0,0){\psline[linewidth=2pt,linestyle=dashed]{->}(0,0)(4,0)}
\rput[br]{0}(0,0){\psline[linewidth=2pt,linestyle=dashed]{->}(0,0)(4,0)}
\rput[B]{0}(3.9,2){\Large{$V_1$}}
\endpspicture
\vspace{-0.5cm}
\caption{Alice's and Bob's measurement bases in terms of polarization.}
\label{basen}
\end{figure}

Lemma~\ref{lemma:distance_key_with_ir} and~\ref{lemma:keyrate} imply that this protocol allows for secure key agreement, as 
stated in the following theorem. 
\begin{theorem}\label{th:main}
Protocol~\ref{prot} achieves a positive secret-key-generation rate as soon as 
the parameter estimation shows an approximation of 
PR boxes with an accuracy exceeding $80\%$ and a correlation of the outputs on input $(0,0)$ higher than $98\%$. 
 There exists an event ${\cal A}$ with probability ${\rm Prob\, }[{\cal A}]=2^{-\Omega(n)}$
such that given ${\cal A}$ does not occur and the protocol is not aborted, then Alice and Bob share a common key  that is 
perfectly secret, where this secrecy is based 
solely on the non-signaling condition.
\end{theorem}

\noindent
The above protocol  also allows for traditional entanglement-based quantum key agreement~\cite{ekert}. Therefore, we have the following. 

\begin{corollary}
Protocol~\ref{prot} allows for efficient information-theoretic key agreement if {\em  quantum OR relativity 
theory\/} is correct.  
\end{corollary}

\section{Concluding Remarks and Open Questions}\label{sec:conclusion}

We propose a new efficient protocol for generating a secret key between two parties
connected by a quantum channel whose security is guaranteed 
solely by the fact that the measured correlations violate a Bell inequality. Quantum mechanics guarantees the protocol to work, i.e., 
the required correlations to occur. But the security proof is completely independent
of quantum mechanics,
once the non-local correlations are established and have been verified by the legitimate partners.

The {\em practical\/} relevance of this fact is that the resulting security is {\em device-independent\/}: We could 
even use devices manufactured by the adversary to do key agreement. 
The {\em theoretical\/} relevance is that the resulting protocol is secure if {\em either relativity or 
quantum theory is correct}. This is in the spirit of modern cryptography's quest to
minimize assumptions under which security can be proven. 

Our scheme requires space-like separation not only between 
events happening on Alice's and Bob's side, but also between events in the 
same laboratory. It is a natural open question whether the space-like-separation 
conditions can be relaxed. For instance, is it sufficient if they hold on one 
of the two sides? Or in one direction among the $n$ events on each side? 
Obviously, the latter would be very easy to guarantee in practice. 

\subsubsection*{Acknowledgments.}
We thank Roger Colbeck, Matthias Fitzi, Severin Winkler, and J\"urg Wullschle\-ger for helpful discussions, 
and Hoi-Kwong Lo for bringing reference~\cite{MayersYao98} to our attention. 
EH and SW are supported by
the Swiss National Science Foundation (SNF) as well as by the ETH
research commission. 
RR acknowledges support from the Swiss National Science Foundation 
(grant No. 200021-119868).

\newpage

\appendix

\section*{\Large{Appendix}}

\section{Best Box Partition of a Single Box}\label{sec:single_box}

In this section, we show that the bound derived in Lemma~\ref{lemma:guessing_probability_single_box} is tight. 

\begin{lemma}\label{lemma:best_partition_single_box}
Assume a box $P_{XYZ|UVW}$ such that the marginal $P_{XY|UV}$ is a non-local box with 
$1/4\cdot \sum_{x\oplus y=u\cdot v}P_{XY|UV}(x,y,u,v)=1-\ep$ 
and $\ep\leq 0.25$. Then there exists a box partition $w$ such that knowing the inputs, $Z$ gives binary erasure information about 
$X$ and $P(Z\in \{0,1\})=4\ep$. This box partition reaches 
\begin{eqnarray}
\nonumber 
d(X|Z(W),Q)=1/2\cdot4\ep
\end{eqnarray}
 for 
 $Q:=(U=u,V=v)$.
\end{lemma}
\begin{proof}
The proof is the following box partition: 
\begin{tiny}
\begin{eqnarray}
\label{eq:1box_decomposition_general}
 \begin{array}{c c||c|c||c|c||}
 $\backslashbox{V}{U}$
 & & \multicolumn{2}{c||}{0} & \multicolumn{2}{c||}{1} \\
  & $\backslashbox{Y}{X}$ 
 & 0 & 1 & 0 & 1 \\ \hline\hline
 \multirow{2}{*}{0} & 0 & a_1 & a_2 & b_1 & b_2 \\ \cline{2-6}
 & 1 & a_2 & a_4 & b_3 & b_4 \\ \hline\hline
 \multirow{2}{*}{1} & 0 & c_1 & c_2 & d_1 & d_2 \\ \cline{2-6}
 & 1 & c_3 & c_4 & d_3 & d_4 \\ \hline\hline
 \end{array}
&=&
a_2\cdot
\begin{array}{c c||c|c||c|c||}
& & \multicolumn{2}{c||}{0} & \multicolumn{2}{c||}{1} \\
 &  
& 0 & 1 & 0 & 1 \\ \hline\hline
\multirow{2}{*}{0} & 0 & 0 & 1 & 1 & 0 \\ \cline{2-6}
& 1 & 0 & 0 & 0 & 0 \\ \hline\hline
\multirow{2}{*}{1} & 0 & 0 & 0 & 0 & 0 \\ \cline{2-6}
& 1 & 0 & 1 & 1 & 0 \\ \hline\hline
\end{array}
+
a_3\cdot
\begin{array}{c c||c|c||c|c||}
& & \multicolumn{2}{c||}{0} & \multicolumn{2}{c||}{1} \\
 & 
& 0 & 1 & 0 & 1 \\ \hline\hline
\multirow{2}{*}{0} & 0 & 0 & 0 & 0 & 0 \\ \cline{2-6}
& 1 & 1 & 0 & 0 & 1 \\ \hline\hline
\multirow{2}{*}{1} & 0 & 1 & 0 & 0 & 1 \\ \cline{2-6}
& 1 & 0 & 0 & 0 & 0 \\ \hline\hline
\end{array}
+
b_2\cdot
\begin{array}{c c||c|c||c|c||}
& & \multicolumn{2}{c||}{0} & \multicolumn{2}{c||}{1} \\
 & 
& 0 & 1 & 0 & 1 \\ \hline\hline
\multirow{2}{*}{0} & 0 & 1 & 0 & 0 & 1 \\ \cline{2-6}
& 1 & 0 & 0 & 0 & 0 \\ \hline\hline
\multirow{2}{*}{1} & 0 & 1 & 0 & 0 & 1 \\ \cline{2-6}
& 1 & 0 & 0 & 0 & 0 \\ \hline\hline
\end{array}
+
b_3\cdot
\begin{array}{c c||c|c||c|c||}
& & \multicolumn{2}{c||}{0} & \multicolumn{2}{c||}{1} \\
 & 
& 0 & 1 & 0 & 1 \\ \hline\hline
\multirow{2}{*}{0} & 0 & 0 & 0 & 0 & 0 \\ \cline{2-6}
& 1 & 0 & 1 & 1 & 0 \\ \hline\hline
\multirow{2}{*}{1} & 0 & 0 & 0 & 0 & 0 \\ \cline{2-6}
& 1 & 0 & 1 & 1 & 0 \\ \hline\hline
\end{array}
\\ \nonumber
&&
+
c_2\cdot
\begin{array}{c c||c|c||c|c||}
& & \multicolumn{2}{c||}{0} & \multicolumn{2}{c||}{1} \\
 & 
& 0 & 1 & 0 & 1 \\ \hline\hline
\multirow{2}{*}{0} & 0 & 0 & 0 & 0 & 0 \\ \cline{2-6}
& 1 & 0 & 1 & 0 & 1 \\ \hline\hline
\multirow{2}{*}{1} & 0 & 0 & 1 & 0 & 1 \\ \cline{2-6}
& 1 & 0 & 0 & 0 & 0 \\ \hline\hline
\end{array}
+
c_3\cdot
\begin{array}{c c||c|c||c|c||}
& & \multicolumn{2}{c||}{0} & \multicolumn{2}{c||}{1} \\
 & 
& 0 & 1 & 0 & 1 \\ \hline\hline
\multirow{2}{*}{0} & 0 & 1 & 0 & 1 & 0 \\ \cline{2-6}
& 1 & 0 & 0 & 0 & 0 \\ \hline\hline
\multirow{2}{*}{1} & 0 & 0 & 0 & 0 & 0 \\ \cline{2-6}
& 1 & 1 & 0 & 1 & 0 \\ \hline\hline
\end{array}
+
d_1\cdot
\begin{array}{c c||c|c||c|c||}
& & \multicolumn{2}{c||}{0} & \multicolumn{2}{c||}{1} \\
 & 
& 0 & 1 & 0 & 1 \\ \hline\hline
\multirow{2}{*}{0} & 0 & 1 & 0 & 1 & 0 \\ \cline{2-6}
& 1 & 0 & 0 & 0 & 0 \\ \hline\hline
\multirow{2}{*}{1} & 0 & 1 & 0 & 1 & 0 \\ \cline{2-6}
& 1 & 0 & 0 & 0 & 0 \\ \hline\hline
\end{array}
+
d_4\cdot
\begin{array}{c c||c|c||c|c||}
& & \multicolumn{2}{c||}{0} & \multicolumn{2}{c||}{1} \\
 &  
& 0 & 1 & 0 & 1 \\ \hline\hline
\multirow{2}{*}{0} & 0 & 0 & 0 & 0 & 0 \\ \cline{2-6}
& 1 & 0 & 1 & 0 & 1 \\ \hline\hline
\multirow{2}{*}{1} & 0 & 0 & 0 & 0 & 0 \\ \cline{2-6}
& 1 & 0 & 1 & 0 & 1 \\ \hline\hline
\end{array}
\\ \nonumber
&&
+
(1-a_2-a_3+b_2-b_3-c_2+c_3+d_1-d_4)\cdot
\begin{array}{c c||c|c||c|c||}
& & \multicolumn{2}{c||}{0} & \multicolumn{2}{c||}{1} \\
 &  
& 0 & 1 & 0 & 1 \\ \hline\hline
\multirow{2}{*}{0} & 0 & \frac{1}{2} & 0 & \frac{1}{2} & 0 \\ \cline{2-6}
& 1 & 0 & \frac{1}{2} & 0 & \frac{1}{2} \\ \hline\hline
\multirow{2}{*}{1} & 0 & \frac{1}{2} & 0 & 0 & \frac{1}{2} \\ \cline{2-6}
& 1 & 0 & \frac{1}{2} & \frac{1}{2} & 0 \\ \hline\hline
\end{array},
\end{eqnarray}
\end{tiny}
To see that (\ref{eq:1box_decomposition_general}) indeed defines a box partition, notice that the parameters $a_2$, $a_3$, $b_2$, $b_3$, $c_2$, $c_3$, $d_1$, $d_4$ (the ones for which the CHSH condition is not fulfilled, i.e., $x\oplus y\neq u\cdot v$) fully characterize any box. By the normalization 
($\sum_i a_i=1$; 
and similar for $b$, $c$ and $d$) 
and non-signaling condition
($a_1+a_2=b_1+b_2$;
and similar for the other rows and columns)
we can express $a_1$ as 
\begin{eqnarray}
\nonumber a_1=\frac{1}{2}\cdot (1-a_2-a_3+b_2-b_3-c_2+c_3+d_1-d_4).
\end{eqnarray}
This shows that the right-hand side and left-hand side of the equation are indeed equal. 
Because we assumed $4\ep\leq 1$, the above decomposition represents a convex combination of several boxes and is, therefore, itself a box. 

With probability $a_2-a_3+b_2-b_3-c_2+c_3+d_1-d_4=4\ep$, $Z$ is such that $P_{XY|UV}^z$ is local deterministic 
(i.e., $X$ ($Y$) is a deterministic function of $U$ ($V$)), in which case knowing $U=u$ and $V=v$, $Z$ gives perfect information 
about $X$ ($Z\in \{0,1\}$). With probability $1-4\ep$ $Z$ is such that $P_{XY|UV}^z$ is a perfect non-local box in which case 
$Z$ cannot give any information about $X$ by the non-signaling condition ($Z=\bot$). 
\pe
\end{proof}

\section{Linear Programming}\label{sec:lin_prog}

In this section, we very briefly state the main facts about linear programming that we 
use for our argument. See, for example,~\cite{bv} for a more detailed introduction.

A \emph{linear program} is an optimization problem with a linear objective function and linear inequality (and equality) constraints, i.e., it can be expressed as 
\begin{eqnarray}
\nonumber \text{max: } &&b^T\cdot x\\
\nonumber  \text{s.t. } &&A\cdot x\leq c
 \ ,
\end{eqnarray}
where $x$ is the variable we want to optimize. An $x$ which fulfills the constraints is called \emph{feasible}. The set of feasible $x$ is convex, more precisely, a convex polyhedron, that is, a convex set with a finite number of extremal points (vertices). A feasible $x$ which maximizes the objective function $b^T\cdot x$, is called \emph{optimal solution} and is denoted by $x^*$. The value of $b^T\cdot x^*$, i.e., the maximal value of the objective function is called \emph{optimal value} and denoted by $q^*$. There is always a vertex at which the optimal value is attained. \\
An important notion of linear programming is duality: the above linear program is called the \emph{primal} problem. From this linear program, another linear program can be derived, defined by 
\begin{eqnarray}
 \nonumber \text{min: } &&c^T\cdot \lambda\\
\nonumber  \text{s.t. } &&A^T\cdot \lambda = b\\
\nonumber && \lambda \geq 0 \ ,
\end{eqnarray}
this problem is called the \emph{dual}, its optimal solution is denoted by $\lambda^*$ and its optimal value by $d^*$. The weak duality theorem says, that the value of the primal objective function for every feasible $x$ is smaller or equal to the value of the dual objective function for every feasible $\lambda$. The strong duality theorem says that the two optimal values are equal, i.e., $q^*=d^*$. It is therefore possible to solve a linear program either by solving the linear program itself, or by solving its dual. 

\begin{center}
\pspicture*[](-1,-1)(7,4)
\psset{unit=0.75cm}
\pspolygon[linewidth=1pt,fillstyle=solid,fillcolor=lightgray](5,-0.5)(6.5,1)(4.5,3)(1,1.5)(1.5,0)
\rput[c]{0}(3,1){\small{feasible region}}
\psline[linewidth=0.5pt, linecolor=gray]{-}(-1,6)(7,2)
\psline[linewidth=0.5pt, linecolor=gray]{-}(-1,5)(7,1)
\psline[linewidth=0.5pt, linecolor=gray]{-}(-1,4)(7,0)
\psline[linewidth=0.5pt, linecolor=gray]{-}(-1,3)(7,-1)
\psline[linewidth=0.5pt, linecolor=gray]{-}(-1,2)(5,-1)
\psline[linewidth=0.5pt, linecolor=gray]{-}(-1,1)(3,-1)
\psline[linewidth=0.5pt, linecolor=gray]{-}(-1,0)(3,-2)
\psline[linewidth=2pt]{->}(-0.5,0)(1.25,3.5)
\rput[c]{63.4349488}(0,2){\small{value of objective function}}
\endpspicture
\end{center}

\section{Explicit Values of the Linear Program for a Single Box}\label{sec:explicite_1box}

In this section, we give the explicit expressions for the parameters 
of the linear program described in Section~\ref{sec:xorpa} for the case of a single box. 

For a single box, $A,b,c$ have the values
\begin{eqnarray}
\nonumber 
 A_1=\left(
\begin{array}{c}
 A_1^{n-s}\\
 -A_1^{n-s}\\
1_{16\times 16}\\
-1_{16\times 16}\\
\end{array}
\right)
\ \text{ with }\ 
A_1^{n-s}=
\left(
\begin{array}{cccccccccccccccc}
 1& 1& -1 &-1 &0 &0 &0 &0 &0 &0 &0 &0 &0 &0& 0& 0\\
    0& 0& 0& 0 &1& 1& -1& -1& 0 &0& 0& 0& 0& 0& 0& 0\\
    0& 0& 0& 0& 0& 0& 0& 0& 1& 1& -1& -1& 0& 0& 0 &0\\
    0& 0& 0& 0& 0& 0& 0& 0& 0& 0& 0& 0& 1& 1& -1& -1\\
    1& 0& 0& 0& 1& 0& 0& 0& -1& 0 &0& 0& -1& 0& 0& 0\\
    0& 1& 0& 0& 0& 1& 0& 0& 0& -1& 0& 0& 0& -1& 0& 0\\
    0& 0& 1& 0& 0& 0& 1& 0& 0& 0& -1& 0& 0& 0& -1& 0\\
    0& 0& 0& 1& 0& 0& 0& 1& 0& 0& 0 &-1& 0& 0& 0& -1 
\end{array}
\right)
\end{eqnarray}

\begin{eqnarray}
\nonumber 
c_1=
\left(
\begin{array}{c}
0_{16}\\
0_{16}\\
P(xy|uv)\\
P(xy|uv)
\end{array}
\right) \, \ \
b_1=
\left(
\begin{array}{c}
 1\\ 1\\ 0\\ 0\\ -1 \\-1 \\0 \\0 \\0 \\0 \\0 \\0 \\0 \\0\\ 0\\ 0\\
\end{array}
 \right)
\ \text{ with }\ 
P(xy|uv)=
\left(
\begin{array}{c}
P(00|00)\\ P(01|00)\\ P(00|01)\\ P(01|01)\\ 
P(10|00)\\ P(11|00)\\ P(10|01)\\ P(11|01)\\ 
P(00|10)\\ P(01|10)\\ P(00|11)\\ P(01|11)\\ 
P(10|10)\\ P(11|10)\\ P(10|11)\\ P(11|11)
\end{array}
 \right)
\end{eqnarray}
 and the dual optimal $\lambda$ is 
\begin{small}
\begin{eqnarray}
\nonumber  \lambda_1^{*T}=
 \left(
\begin{array}{cccccccccccccccccccccccccccccccccccccccccccccccccccccccccccccccc}
0.5 & 0 &  0.5 & 0 &  0.5& 0 & 0.5 & 0 & 0 & 0.5 & 0 & 0.5 & 0 & 0.5 & 0 & 0.5 &
    0 & 1 & 0 & 1 & 0 & 0 & 0 & 0 & 0 & 0 & 1& 0 & 1 & 0 & 0 & 0 & 0 & 0 & 0 & 0 & 1 & 0 & 1 & 0 & 0 & 1 & 0 & 0 & 0 & 0 & 0 & 1
\end{array}
\right)
\end{eqnarray}
\end{small}

To obtain the value of the objective function ($c^T\cdot \lambda_1^*$), the first part of $\lambda_1^*$ will be multiplied by $0$, i.e., does not contribute to the value. The second part is multiplied with $P_{XY|UV}$. We can easily see by comparison that for every $x,y,u,v$ such that $x\oplus y\neq u\cdot v$ there is exactly one $1$ in the second part of $\lambda_1^*$ and everywhere else $\lambda_1^*$ is $0$. I.e., 
\begin{eqnarray}
 \nonumber c^T\cdot \lambda_1^*=\sum_{x,y,u,v:x\oplus y\neq u\cdot v}P_{XY|UV}(x,y,u,v)
\end{eqnarray}

The above values are for the input $u,v=0,0$. The optimal $\lambda^*$ reaching the same value for different $u,v$ are given below:\\
For $u,v=0,1$:
\begin{small}
\begin{eqnarray}
\nonumber b_1^T&=&
\left(
\begin{array}{cccccccccccccccc}
0& 0& 1& 1& 0& 0& -1&-1 &0 &0 &0 &0 &0 &0 &0 &0 
\end{array}
 \right)\\
\nonumber 
\lambda_1^{*T}&=&
 \left(
\begin{array}{cccccccccccccccccccccccccccccccccccccccccccccccccccccccccccccccc}
0 & 0.5 & 0.5 & 0 & 0.5 & 0 & 0.5 & 0 & 0.5 & 0 & 0 & 0.5 & 0 & 0.5 & 0 & 0.5 & 0 & 1 & 0 & 1 & 0 & 0 & 0 & 0 & 0 & 0 & 1 & 0 & 1 & 0 & 0 & 0 & 0 & 0 & 0 & 0 & 1 & 0 & 1 & 0 & 0 & 1 & 0 & 0 & 0 & 0 & 0 & 1
\end{array}
\right)
\end{eqnarray}
\end{small}
For $u,v=1,0$:
\begin{small}
\begin{eqnarray}
\nonumber b_1^T&=&
\left(
\begin{array}{cccccccccccccccc}
0 &0 &0& 0& 1& 1& 0& 0& -1&-1 &0 &0 &0 &0 &0 &0 
\end{array}
 \right)\\
\nonumber 
\lambda_1^{*T}&=&
 \left(
\begin{array}{cccccccccccccccccccccccccccccccccccccccccccccccccccccccccccccccc}
0.5 & 0 & 0.5 & 0 & 0 & 0.5 & 0.5 & 0 & 0 & 0.5 & 0 & 0.5 & 0.5 & 0 & 0 & 0.5 & 0 & 0 & 0 & 1 & 1 & 0 & 0 & 0 & 0 & 1 & 1 & 0 & 0 & 0 & 0 & 0 & 0 & 1 & 0 & 0 & 0 & 0 & 1 & 0 & 0 & 0 & 0 & 0 & 1 & 0 & 0 & 1
\end{array}
\right)
\end{eqnarray}
\end{small}
For $u,v=1,1$:
\begin{small}
\begin{eqnarray}
\nonumber b_1^T&=&
\left(
\begin{array}{cccccccccccccccc}
0& 0& 1& 1& 0& 0& -1&-1 &0 &0 &0 &0 &0 &0 &0 &0 
\end{array}
 \right)\\
\nonumber 
\lambda_1^{*T}&=&
 \left(
\begin{array}{cccccccccccccccccccccccccccccccccccccccccccccccccccccccccccccccc}
0.5 & 0 & 0 & 0.5 & 0 & 0.5 & 0.5 & 0 & 0 & 0.5 & 0.5 & 0 & 0.5 & 0 & 0 & 0.5 & 0 & 0 & 0 & 1 & 1 & 0 & 0 & 0 & 0 & 1 & 1 & 0 & 0 & 0 & 0 & 0 & 0 & 1 & 0 & 0 & 0 & 0 & 1 & 0 & 0 & 0 & 0 & 0 & 1 & 0 & 0 & 1
\end{array}
\right)
\end{eqnarray}
\end{small}

\section{Parameter Estimation}\label{sec:parameter_estimation}

A crucial step of any quantum key distribution protocol is parameter estimation. 
Alice and Bob need to test a small sample of the boxes they have received, to see whether they have received boxes with the correct parameters.
This can be done by classical sampling theory, as given in~\cite{KoeRen04b} (see also~\cite{HHHLO}). 
\begin{lemma}{\bf \cite{KoeRen04b},\cite{HHHLO}}
Let $Z$ be an $n$-tuple and $Z'$ a $k$-tuple of random variables over a set $\mathcal{Z}$, with symmetric joint probability $P_{ZZ'}$. Let $Q_{z'}$ be the relative frequency distribution of a fixed sequence $z'$ and $Q_{(z,z')}$ be the relative frequency distribution of a sequence $(z,z')$, drawn according to $P_{ZZ'}$. Then for every $\ep\geq 0$ we have
\begin{eqnarray}
\nonumber  P_{ZZ'}[||Q_{(z,z')}-Q_{z'}||\geq \ep]\leq |\mathcal{Z}|\cdot e^{-k\ep^2/8|\mathcal{Z}|}
\end{eqnarray}
\end{lemma}
In our case, we consider the case when Alice and Bob share $n+k$ 
boxes. After they have used the boxes and announced the inputs, they randomly choose $k$ of the boxes, for which they also uncover the outputs. Call $\ep_{meas}$ the fraction of those $k$ boxes which $x\oplus y\neq u\cdot v$. We call $\bar{\ep}$ the average error of the remaining boxes. Then,
\begin{eqnarray}
\nonumber  P_{S}[||\frac{n}{n+k}(\bar{\ep}-\ep_{meas})||\geq \ep]\leq 2\cdot e^{-k\ep^2/16}\\
\nonumber
P_{S}[\bar{\ep}\geq \ep_{meas} (1+\frac{k}{n})\cdot\ep]\leq 2\cdot e^{-k\ep^2/16}
\end{eqnarray}
Obviously, Alice and Bob can also test other parameter such as $\delta$ --- the correlation of their output bits given input $(0,0)$ --- in a similar way.

This means, if the boxes Eve has distributed are not good enough for key agreement, Alice and Bob will most certainly detect this. If they are good enough, then Alice's and Bob's test will most certainly be passed and key agreement is possible as discussed above.

\section{Depolarization}\label{sec:depol}

Assume Alice and Bob share an arbitrary distribution 
$P_{\bof{XY}|\bof{UV}}$ where $\bof{X},\bof{Y},\bof{U},\bof{V}$ is an n-bit string. 
Then they can perform a sequence of local operations and 
public communication in order to obtain a distribution which 
corresponds to the convex combination of $n$ independent 
approximations of a PR box with error $\ep_i$. Further, each 
approximation of the PR box $P_{X_iY_i|U_iV_i}$ has unbiased 
outcomes and the same error $\ep_i$ for all inputs. The local 
operations achieving this, are given in~\cite{mag,mrwbc}. 
We restate them here briefly: For each $i$, Alice and Bob choose the mapping independently in two steps. 
First, with probability $1/2$, they do either of the following:
\begin{enumerate}
 \item nothing
 \item both flip their outcome bits, i.e., $x_i\rightarrow x_i\oplus 1$ and $y_i\rightarrow y_i\oplus 1$\ .
\end{enumerate}
Then, with probability $1/4$ each, they do either of the following:
\begin{enumerate}
 \item nothing
 \item $x_i\rightarrow x_i\oplus u_i$ and $v_i\rightarrow v_i\oplus 1$
 \item $u_i\rightarrow u_i\oplus 1$ and  $y_i\rightarrow y_i\oplus v_i$ 
 \item $u_i\rightarrow u_i\oplus 1$, $x_i\rightarrow x_i\oplus u_i\oplus 1$, $v_i\rightarrow v_i\oplus 1$ and $y_i\rightarrow y_i\oplus v_i$\ .
\end{enumerate}
The choice of local operation needs $3$ random bits per box which have to be communicated from Alice to Bob. Because, each of these operations conserves the probability of error $\ep_i$ a box with the same error parameter --- but now an unbiased one with the same error for all inputs --- is obtained. Furthermore, when this transformation is applied to each input/output bit of a distribution $P_{XY|UV}$ taking $n$ bits input and giving $n$ bits output, a convex combination of products of
independent and unbiased approximations of PR boxes (with possibly different error $\ep_i$) 
is obtained.

\section{Eve Can Always Know a Certain Fraction of Bits}\label{sec:better_collective_attack}

Can Eve really do collective attacks which are better than individual ones? In this section we show that this is indeed possible and give a 
collective attack for the case when Alice and Bob share $n$ boxes with error $\ep$ and such that the error is the same for all 
inputs.\footnote{In the following, we will only consider unbiased boxes with the same error for all inputs. The box is, therefore, fully 
characterized by its error $\ep$.} 
We show that for every value of $\ep$ there exists an attack of Eve such that she knows with certainty a fraction of all the output 
bits of Alice --- an option unavailable if only individual 
attacks are allowed. What fraction Eve can know depends on the value of $\ep$. 

\subsection{Example of a Better Collective Attack for Two Boxes}

We first describe an example of an attack on two boxes. 
We will give an explicit strategy of Eve (a box partition) which shows that she can know either one of the two bits with higher probability 
than what can be done by an individual attack (although Eve cannot choose which one of the two bits she will get to know). This shows 
that collective attacks are 
strictly stronger 
than individual attacks. In fact, assume Alice will communicate to Bob the XOR of her two output bits in the information reconciliation 
phase. In that case only the probability that Eve knows at least one of the two bits is important, because together with the information of 
the XOR this immediately gives her full information about \emph{both} bits. 

Before we can give the box partition, we need to proof the following Lemma. 
\begin{lemma}
Every box with 
$\ep \in [1/4,3/4]$ 
is 
local and can be expressed as the convex combination of local deterministic boxes. 
We use the short-hand notation 
$L_{\ep}$ 
for these local 
$\ep$-boxes. 
\end{lemma}
\begin{proof}
According to Lemma~\ref{lemma:best_partition_single_box}, a 
box with $\ep=0.25$ is 
 local
and can be expressed as convex combination of local deterministic boxes. 
A 
local  box with $\ep=0.75$ can be obtained from the one with $\ep=0.25$ by flipping one of the output bits. Every box with $\ep \in (1/4,3/4)$ 
 can then be expressed as a convex combination of the above boxes and is therefore 
local. 
\pe
\end{proof}
We have already seen that if a box can be expressed as convex combination of local deterministic boxes, then there exists a box 
partition (where the elements are exactly the local deterministic boxes) such that knowing the inputs, the outputs are completely determined. 
CHSH-game with a local box, we will now see that in our example the \emph{bad} (local) strategies are important. 

Eve's strategy is given by Lemma~\ref{lemma:better_collective_attack_for_two_boxes}. Note that the local boxes with 
the \emph{largest} error  play an important role here. Eve's outcome $Z$ composed of two symbols, such that the first describes the first 
box given outcome $Z=z$ and the second symbol the second box. More precisely, we use $z_i=l$ for an outcome given which box $i$ is local 
and $z_i=\bot$ for an outcome given which box $i$ is a PR box. 
\begin{lemma}\label{lemma:better_collective_attack_for_two_boxes}
Assume a $(2\cdot2+1)$-partite non-signaling distribution $P_{\bof{XY}Z|\bof{UV}W}$ such that $P_{\bof{XY}|\bof{UV}}$ corresponds 
to two independent 
boxes with $P(X_i\oplus Y_i=U_i\cdot V_i)=1-\ep$ for all inputs and $i=1,2$. Then the following is a box partition:
\begin{eqnarray}
\label{eves_strategy_for_two_boxes} \begin{array}{c|c|c|c}
Z_1Z_2 & p^{z_1z_2} & P^{z_1z_2}_{X_1Y_1|U_1V_1} &  P^{z_1z_2}_{X_2Y_2|U_2V_2}\\ \hline
ll & (4/3\ep)^2 & L_{\ep=3/4} & L_{\ep=3/4} \\
l\bot & (4\ep)(1-4/3\ep) & L_{\ep=1/4} & NL \\
\bot l & (1-4/3\ep)(4\ep)  & NL & L_{\ep=1/4} \\
\bot \bot & 1-2\cdot(4\ep)(1-4/3\ep)-(4/3\ep)^2 & NL & NL \\
\end{array}
\end{eqnarray}
and $P_{\bof{XY}|\bof{UV}}^{z_1z_2}:=P^{z_1z_2}_{X_1Y_1|U_1V_1}\cdot P^{z_1z_2}_{X_2Y_2|U_2V_2}$ and 
where $L_{\ep}$ stands for a box with error $\ep$ and $NL$ for a PR box. 
\end{lemma}
\begin{proof}
To see that this defines a box partition, let us first see, that all boxes given outcome $z$ are non-signaling between 
all four input/output ends. This is obviously the case, because each of the two boxes given outcome $z$ is non-signaling and the 
double-box given outcome $z$ is given by the product of the two individual boxes. \\
Now let us see that the marginal is correct. For this, we need to verify that the distribution of the output bits on each side is correct, 
but also that the probability that any subset of boxes fulfills the CHSH condition needs to be correct. The first condition is fulfilled 
because all output bits are uniform, independent and random even given outcome $z$. Now let us see that the probability to fulfill/violate 
the CHSH condition is also correct. The probability that both boxes violate the CHSH condition is given by the probability to obtain 
$z_1z_2=ll$ (both boxes are local) times the probability that they then violate the CHSH condition. (If either of the boxes given outcome 
$z$ is a non-local box it never violates the CHSH condition, therefore no other outcomes $z$ have to be considered.) 
\begin{eqnarray}
\nonumber
 P(X_i\oplus Y_i\neq U_i\cdot V_i\  \text{for}\ i=1,2)&=& 
 p^{ll}\cdot \ep_{box\ 1}^{z_1z_2=ll}\cdot \ep_{box\ 2}^{z_1z_2=ll}
 =
 (4/3\ep)^2\cdot (3/4)^2=\ep^2\ ,
\end{eqnarray}
where $\ep_{box\ 1}^z$ denotes the error of the first box given outcome $Z=z$.
Similarly, we can also show that the probability of the first box violating the CHSH condition is correct:
\begin{eqnarray}
\nonumber
 P((X_1\oplus Y_1\neq U_1\cdot V_1))&=& 
  p^{\{0,1\}^2}\cdot \ep_{box\ 1}^{z_1z_2=ll}+
   p^{l\bot}\cdot \ep_{box\ 1}^{z_1z_2=l\bot} \\
\nonumber &=&
 (4/3\ep)^2\cdot (3/4)+(4\ep)(1-4/3\ep)\cdot (1/4)=\ep,
\end{eqnarray}
and the same for all other subsets of boxes. This shows that the marginal $P_{\bof{XY}|\bof{UV}}$ is unchanged 
by this box partition. \pe
\end{proof}
From this box partition, we directly obtain as a corollary:
\begin{corollary}
Assume a $(2\cdot2+1)$-partite non-signaling distribution $P_{\bof{XY}Z|\bof{UV}W}$ such that $P_{\bof{XY}|\bof{UV}}$ corresponds 
to two independent boxes with $P(X_i\oplus Y_i=U_i\cdot V_i)=1-\ep$ for $i=1,2$. Then there exists a box partition $w$ such that 
the probability that $Z$ gives binary erasure information (knowing $\bof{U}=\bof{u},\bof{V}=\bof{v}$) about at least one of the two 
output bits $X_1$, $X_2$ is $(4\ep)^2+2\cdot(4\ep)(1-4/3\ep)$.
\end{corollary}
This is larger than $(4\ep)^2+2\cdot(4\ep)(1-4\ep)$, the value obtained by the best individual strategy.

\subsection{Better Collective Attack for Any Number of Boxes and $\ep$}

We now give a generalization of the above strategy to attack two boxes to any number of boxes. The attack obtains knowledge about a 
fraction of the output with \emph{certainty} and independently of the total number of boxes. Which fraction can be known depends on 
the error of the boxes. 
\begin{lemma}\label{lemma:better_collective_attack_for_n_boxes}
Assume a $(2n+1)$-partite non-signaling distribution $P_{\bof{XY}Z|\bof{UV}W}$ such that $P_{\bof{XY}|\bof{UV}}$ corresponds to 
$n$ boxes with $P(X_i\oplus Y_i=U_i\cdot V_i)=1-\ep$ for $i=1,\ldots ,n$. Then the following is a box partition:
\begin{eqnarray}
\label{eves_strategy_for_n_boxes} \begin{array}{c|c|c}
Z & p^{z} & P^{z}_{\bof{XY}|\bof{UV}} \\ \hline
\{z| \sharp l=i \in [2,n]\} & (4/3\ep)^i(1-4/3\ep)^{n-i} & (L_{\ep=3/4})^i\cdot (NL)^{n-i} \\
\{z|\sharp l=1\}  & (4\ep)(1-4/3\ep)^{n-1} & L_{\ep=1/4}\cdot (NL)^{n-1}  \\
\{z|\sharp l=0\}   & 1-\sum_{z|\sharp l\geq 1}p^z & (NL)^n
\end{array}
\end{eqnarray}
where $Z$ is composed of $n$ symbols ($\{l,\bot \}^n$) and we write $\sharp l=i$ for a $z$ which contains $i$ symbols $l$. 
\end{lemma}
The proof is analogue to the proof of Lemma~\ref{lemma:better_collective_attack_for_two_boxes}.
From the above box partition, we obtain the following lemma. 
\begin{lemma}
Assume a $(2n+1)$-partite non-signaling distribution $P_{\bof{XY}Z|\bof{UV}W}$ such that $P_{\bof{XY}|\bof{UV}}$ corresponds to 
$n$ independent boxes with $P(X_i\oplus Y_i=U_i\cdot V_i)=1-\ep$.
Then, whenever $\ep\geq \frac{3}{8\cdot n+4}$,
there exists a box partition $w$ such that for all outcomes $z$ $P_{\bof{XY}|\bof{UV}}^z$ is such that at least one of the 
$n$ boxes is fully local. 
\end{lemma}
\begin{proof}
We use the box partition given in Lemma~\ref{lemma:better_collective_attack_for_n_boxes}. 
The probability to obtain an outcome $Z$ such that at least $1$ of the $n$ boxes given $Z$ is fully local can be expressed as 
 \begin{eqnarray}
\nonumber
\sum_{\{z|\sharp l=i\geq 1\} }p^z&=& \sum_{i=n}^{2}\binom{n}{i}(4/3\ep)^{n-i}(1-4/3\ep)^i+\binom{n}{1}(4\ep)(1-4/3\ep)^{n-1}
\end{eqnarray}
Because of the binomial formula, this probability is equal to $1$ whenever
\begin{eqnarray}
\nonumber
n\cdot (4\ep)(1-4/3\ep)^{n-1}&=&n\cdot (4/3\ep)(1-4/3\ep)^{n-1}+(1-4/3\ep)^{n}
\end{eqnarray}\peon
\end{proof}
Therefore, whenever $\ep\geq \frac{3}{8\cdot n+4}$, Eve can know $1$ of the $n$ bits with certainty. 
Or said differently, Eve can know roughly a fraction of 
$f=1/n=\frac{8\ep}{3-4\ep}\geq 8\ep/3$ of the bits with certainty.

\section{Proofs}

\subsection{All Non-Signaling Conditions}\label{sec:imply_ns}

In this section, we show that Condition~\ref{condition:ns}' implies the non-signaling 
condition between all possible subsets of interfaces of the box. 
\begin{lemma}
 Assume a system $P_{XYZ|UVW}$ such that 
\begin{eqnarray}
\nonumber \sum\nolimits_x P_{XYZ|UVW}(x,y,z,u,v,w)&=&\sum\nolimits_xP_{XYZ|UVW}(x,y,z,u',v,w)\ \forall y,z,v,w\\
\nonumber \sum\nolimits_yP_{XYZ|UVW}(x,y,z,u,v,w)&=&\sum\nolimits_yP_{XYZ|UVW}(x,y,z,u,v',w)\ \forall x,z,u,w\\
\nonumber \sum\nolimits_zP_{XYZ|UVW}(x,y,z,u,v,w)&=&\sum\nolimits_zP_{XYZ|UVW}(x,y,z,u,v,w')
\ \forall x,y,u,v
\end{eqnarray}
Then it also holds that
\begin{eqnarray}
\nonumber \sum\nolimits_{xy} P_{XYZ|UVW}(x,y,z,u,v,w)&=&\sum\nolimits_{xy}P_{XYZ|UVW}(x,y,z,u',v',w)\ \forall z,w\ .
\end{eqnarray}
\end{lemma}
\begin{proof}
\begin{eqnarray}
 \nonumber \sum\nolimits_{xy} P_{XYZ|UVW}(x,y,z,u,v,w)&=&\sum\nolimits_{x}\sum\nolimits_{y} P_{XYZ|UVW}(x,y,z,u,v,w)\\
\nonumber =\sum\nolimits_{x}\sum\nolimits_{y} P_{XYZ|UVW}(x,y,z,u,v',w)&=& \sum\nolimits_{y}\sum\nolimits_{x} P_{XYZ|UVW}(x,y,z,u,v',w)\\
\nonumber = \sum\nolimits_{y}\sum\nolimits_{x} P_{XYZ|UVW}(x,y,z,u',v',w) &=& \sum\nolimits_{xy} P_{XYZ|UVW}(x,y,z,u',v',w)
 \end{eqnarray}
\peon
\end{proof}

\subsection{Distance of Set given other Sets}\label{sec:distance_set_given_other_sets}

The following lemma is used in the proof of Lemma~\ref{lemma:distance_k_th_bit}. 
\begin{lemma}\label{lemma:distance_set_given_other_sets}
Assume random bits $S_1,\ldots,S_k$. If $S_k$ is uniform given all linear combinations over $GF(2)$ of $S_1,\ldots,S_{k-1}$, i.e., 
$P_{S_k|\bigoplus_{i\in I}}(0)=P_{S_k|\bigoplus_{i\in I}}(1)$ for all $I\subseteq \{1,\ldots,k-1\}$, then $S_k$ is uniform given 
$S_1,\ldots,S_{k-1}$, i.e., $P_{S_k|S_1\ldots,S_{k-1}}(0)=P_{S_k|S_1\ldots,S_{k-1}}(1)$.
\end{lemma}
\begin{proof}
We proof the case $k=3$, the general case follows by induction. We have to show that if $P_{S_3|S_1}$, $P_{S_3|S_2}$ and 
$P_{S_3|S_1\oplus S_2}$ are uniform, then $P_{S_3|S_1,S_2}$ is uniform. Consider the probabilities $P_{S_1,S_2,S_3}$. Because 
$P_{S_3|S_1}$ is uniform, we obtain the constraints
\begin{eqnarray}
\label{eq:1} P_{S_1,S_2,S_3}(0,0,0)+P_{S_1,S_2,S_3}(0,1,0)&=&P_{S_1,S_2,S_3}(0,0,1)+P_{S_1,S_2,S_3}(0,1,1)\\
\nonumber P_{S_1,S_2,S_3}(1,0,0)+P_{S_1,S_2,S_3}(1,1,0)&=&P_{S_1,S_2,S_3}(1,0,1)+P_{S_1,S_2,S_3}(1,1,1)\ .
\end{eqnarray}
Because $P_{S_3|S_2}$ is uniform, 
\begin{eqnarray}
\nonumber P_{S_1,S_2,S_3}(0,0,0)+P_{S_1,S_2,S_3}(1,0,0)&=&P_{S_1,S_2,S_3}(0,0,1)+P_{S_1,S_2,S_3}(1,0,1)\\
\label{eq:2} P_{S_1,S_2,S_3}(0,1,0)+P_{S_1,S_2,S_3}(1,1,0)&=&P_{S_1,S_2,S_3}(0,1,1)+P_{S_1,S_2,S_3}(1,1,1)\ .
\end{eqnarray}
And from the fact that  $P_{S_3|S_1\oplus S_2}$ is uniform, we obtain
\begin{eqnarray}
\label{eq:3} P_{S_1,S_2,S_3}(0,0,0)+P_{S_1,S_2,S_3}(1,1,0)&=&P_{S_1,S_2,S_3}(0,0,1)+P_{S_1,S_2,S_3}(1,1,1)\\
\nonumber P_{S_1,S_2,S_3}(0,1,0)+P_{S_1,S_2,S_3}(1,0,0)&=&P_{S_1,S_2,S_3}(0,1,1)+P_{S_1,S_2,S_3}(1,0,1)\ .
\end{eqnarray}
Then substract (\ref{eq:2}) from (\ref{eq:1}) and add (\ref{eq:3}) to obtain
\begin{eqnarray}
2\cdot P_{S_1,S_2,S_3}(0,0,0)&=&2\cdot P_{S_1,S_2,S_3}(0,0,1)\ ,
\end{eqnarray}
which implies 
\begin{eqnarray}
P_{S_3|S_1=0,S_2=0}(0)&=&\frac{P_{S_1,S_2,S_3}(0,0,0)}{P_{S_1,S_2,S_3}(0,0,0)+P_{S_1,S_2,S_3}(0,0,1)}=P_{S_3|S_1=0,S_2=0}(1)\ .
\end{eqnarray}
Uniformity for all other values of $S_1,S_2$ then follows directly from the above equations. 
\pe
\end{proof}

\subsection{Linear Combination of Random Vectors}\label{sec:lin_com_of_random_vect}

The following lemma is used in the proof of Lemma~\ref{lemma:distance_k_th_bit}. 
\begin{lemma}\label{lemma:lin_com_of_random_vect}
Assume $\bof{u}$ and $\bof{v}$ are $n$-bit vectors and $P_U$ is the uniform distribution over all these vectors. 
Define the vector $\bof{w}=\bof{u}\oplus \bof{v}$. Then $\bof{w}$ is again distributed according to the uniform distribution. 
\begin{eqnarray}
\nonumber P_{\bof{u}\leftarrow P_U}P_{\bof{v}\leftarrow P_U}(\bof{u}\oplus \bof{v})&=&P_{\bof{w}\leftarrow P_U}(\bof{w})\ .
\end{eqnarray}
\end{lemma}
\begin{proof}
The uniform distribution over all $n$-bit vectors can be obtained by drawing each of the $n$-bits at random, 
i.e., $P(0)=P(1)=1/2$. The XOR of two random bits is again a random bit, i.e., $P(0)=P(1)=1/2$ and therefore, 
$\bof{w}$ is also a vector drawn according to the uniform distribution over all $n$-bit vectors. \pe
\end{proof}

\subsection{Average Epsilon}\label{subsec:ep_average}
The following lemma is used in the proof of Lemma~\ref{lemma:distance_k_th_bit}. 
\begin{lemma}\label{lemma:ep_average}
Assume a variable $\ep_i$ for $i=1,\ldots ,n$ with average $\ep=\frac{1}{n}\cdot \sum_i \ep_i$. Then
\begin{eqnarray}
\nonumber \sum_{K\subseteq n}\prod_{i\in K}\ep_i
&\leq& \sum_{K\subseteq n}\prod_{i\in K}\ep.
\end{eqnarray}
\end{lemma}
\begin{proof}
We will show that when replacing $\ep_1$ and $\ep_2$ by their average, the value of the above expression only gets bigger. The lemma then follows by repeating to combine $\ep_i$ in pairs and replacing them by their average. First note that replacing $\ep_1$ and $\ep_2$ by $\frac{\ep_1+\ep_2}{2}$ each does not change the average $\ep$. Now calculate 
$\sum_{K\subseteq n}\prod_{i\in K}(4\ep_i)$. For this, divide the sets $K$ into different categories: The ones which contain neither $\ep_1$ nor $\ep_2$, which we call $K^\emptyset$; the ones which contain either $\ep_1$ or $\ep_2$ which we call $K^{\ep_1}$ ($K^{\ep_2}$); and the ones which contain both $\ep_1$ and $\ep_2$ called $K^{\ep_1\ep_2}$. 
\begin{eqnarray}
\nonumber \sum_{K\subseteq n}\prod_{i\in K}\ep_i=\sum_{K^\emptyset}\prod_{i\in K}\ep_i+ \sum_{K^{\ep_1}}\prod_{i\in K}\ep_i
+ \sum_{K^{\ep_2}}\prod_{i\in K}\ep_i+\sum_{K^{\ep_1\ep_2}}\prod_{i\in K}\ep_i= (1+\ep_1+\ep_2+\ep_1\ep_2)\cdot \sum_{K^\emptyset}\prod_{i\in K}\ep_i. 
\end{eqnarray}
When replacing $\ep_1$ and $\ep_2$ by $\frac{\ep_1+\ep_2}{2}$ each, clearly $\sum_{K^\emptyset}\prod_{i\in K}\ep_i$ stays the same and the value of $1+\ep_1+\ep_2+\ep_1\ep_2$ only becomes larger because $\ep_1\ep_2\leq \left(\frac{\ep_1+\ep_2}{2}\right)^2$. 
\pe
\end{proof}

\bibliography{space-like_separated_protocol}
\bibliographystyle{hplain}

\end{document}